\newif\ifcomments\commentsfalse
\newif\ifanon\anonfalse
\newif\iffullversion\fullversiontrue
\newif\ifecon\econfalse

\iffullversion
    \documentclass[12pt,a4paper]{article}
    \usepackage{fullpage}
    \usepackage[T1]{fontenc}
\else
    \documentclass[runningheads]{llncs}
    \usepackage[T1]{fontenc}
\fi

\usepackage{tikz}
\usepackage{graphicx} \usepackage{amsfonts}
\usepackage{amsmath}
\usepackage{amssymb}
\usepackage{bm}
\usepackage{subcaption} \usepackage{enumitem}
\usepackage{multirow}
\usepackage{booktabs}
\usepackage[mathscr]{euscript}

\usepackage[hyphens]{url}
\usepackage[colorlinks,urlcolor=MidnightBlue,citecolor=black,linkcolor=BlueViolet]{hyperref}
\usepackage{breakurl}
\usepackage{algorithm}
\usepackage{algpseudocode}

\iffullversion\else
\usepackage{color}
    
\fi

\usepackage{xspace}

\usepackage[capitalise]{cleveref}

\iffullversion
    \usepackage{amsthm}
    \newtheorem{theorem}{Theorem}
    \newtheorem{lemma}{Lemma}
    \newtheorem{observation}{Observation}
    
    \theoremstyle{remark}
    \newtheorem{remark}{Remark}

\else
\spnewtheorem{observation}{Observation}{\bfseries}{\rmfamily}

\fi
\newcommand{\qqed}{\iffullversion\else\qed\fi}

\newtheorem{innercustomthm}{Theorem}
\newenvironment{customthm}[1]
  {\renewcommand\theinnercustomthm{#1}\innercustomthm}
  {\endinnercustomthm}

\iffullversion
  \crefname{section}{Section}{Sections}
  \crefname{appendix}{Appendix}{Appendices}
  \crefname{figure}{Figure}{Figures}
  \crefname{definition}{Definition}{Definitions}
  \crefname{observation}{Observation}{Observations}
  \crefname{remark}{Remark}{Remarks}
\else
  \crefname{section}{\S}{\S\S}
  \crefname{appendix}{App.}{Apps.}
  \crefname{figure}{Fig.}{Figs.}
  \crefname{definition}{Def.}{Defs.}
  \crefname{observation}{Observation}{Observations}
\fi

\usepackage[dvipsnames]{xcolor}
\usepackage{ulem}
\renewcommand{\emph}[1]{\textit{#1}}

\renewcommand{\phi}{\varnothing}
\newcommand{\eps}{\mathsf{e}}
\newcommand{\Poisson}{\mathrm{Poisson}}
\newcommand{\One}{\mathbf{1}}

\newcommand{\customskip}{\iffullversion\medskip\else\smallskip\fi}
\newcommand{\subh}[1]{\customskip \noindent \textit{{#1}}}
\renewcommand{\paragraph}[1]{\subh{#1.}}
\newcommand{\qparagraph}[1]{\subh{#1?}}

\usepackage[maxcitenames=3, maxbibnames=30]{biblatex}
\def\blu#1{\textcolor{blue}{#1}}

\def\e{equilibrium\xspace}
\def\a{\sigma}

\DeclareMathOperator*{\E}{\mathbb{E}} \renewcommand{\Re}{\mathbb{R}}
\renewcommand{\O}{\mathcal{O}}
\newcommand{\C}{\mathcal{C}}
\newcommand{\OSet}{\bm{\mathbf{O}}} 
\newcommand{\CSet}{\bm{\mathbf{C}}}
\newcommand{\OCSet}{\{\OSet,\CSet\}}

\newcommand{\sys}{\mathcal{T}}
\newcommand{\syss}{\bm{\mathbf{T}}}
\newcommand{\G}{\mathsf{G}}
\newcommand{\B}{\mathsf{B}}
\newcommand{\BG}{\{\B,\G\}}
\newcommand{\Win}{\mathsf{Win}}
\newcommand{\WWin}{\mathsf{BadWin}}

\newcommand{\nature}{ethical type\xspace}

\ifcomments
    \newcommand{\ip}[1]{{\textcolor{Purple}{I: #1}}}
    \newcommand{\sunoo}[1]{{\textcolor{WildStrawberry}{S: #1}}}
\else
    \newcommand{\ip}[1]{\ignorespaces}
    \newcommand{\sunoo}[1]{\ignorespaces}
\fi

\title{The False Allure of ``Security by Obscurity'': \\ A Game-Theoretic Approach}
\iffullversion\else
    \titlerunning{The False Allure of ``Security by Obscurity''}
\fi
\ifanon
    \author{Anonymized for submission}
\else
    \iffullversion
        \author{In-Uck Park\\University of Bristol \and Sunoo Park\\New York University}
    \else
\author{In-Uck Park\inst{1}\orcidID{1111-2222-3333-4444} \and Sunoo Park\inst{2}\orcidID{0000-0002-1884-9585}}
        \institute{University of Bristol, Bristol, UK \and New York University, New York, NY, USA}
        \authorrunning{S. Park and I. Park}
    \fi
\fi
\date{}

\begin{document}

\maketitle

\begin{abstract}
“Security by obscurity,” or the idea that hiding how a system works helps protect it from adversaries,
has long been poorly regarded in the security community. Yet its intuitive appeal persists, and it has resurfaced repeatedly in industry practices and public debates over the years.
The logic, on the one hand, is that open scrutiny of systems is the best way to find and respond to  problems; and on the other hand, that too much openness can increase risk by giving adversaries a roadmap for attack.
This isn't purely a technical question, but a human one: it’s about how attackers and good-faith security researchers respond to openness in the real world. 
As such, security experts base their stance on decades of community experience; however, the field lacks empirical or theoretical foundations to back up these long-standing beliefs.
We provide the first rigorous modeling of ``security by obscurity,'' aiming to capture how rational actors will respond to openness or obscurity in fallible systems. We prove that in equilibrium under realistic model conditions, ``obscured'' systems are more likely to suffer successful attacks than ``open'' systems ---
establishing theoretical support for the longstanding received wisdom.
\end{abstract}

\section{Introduction}\label{sec:intro}

\begin{quote}
    ``Hiding security vulnerabilities in algorithms, software and/or hardware decreases the likelihood that they will be repaired and increases the likelihood that they can and will be exploited by evil-doers.'' \cite{Bellovin2002Obscurity}
\end{quote}

Such is the prevailing consensus within the security and cryptography communities, which have argued against the practice of ``security by obscurity'' since decades ago. Related academic arguments can be traced at least as far back as Kerckhoffs's Principle in 1883, long before modern computers existed, which advised that military cipher systems ``should not require secrecy, and should be able to be stolen by the enemy without causing trouble''\footnote{In the original French: ``Il faut qu’il n’exige pas le secret, et qu’il puisse sans inconvénient tomber entre les mains de l’ennemi''~\cite[vol. 1, p. 12]{Kerckhoffs1883crypto}. \iffullversion
English translation by Fabien Petitcolas~\cite{petitcolas-kerck}.\fi}
\cite{Kerckhoffs1883crypto,petitcolas-kerck}. In modern times, the idea that secrecy enhances security has been decried as a ``fallacy'' and a ``myth'' \cite{Bellovin2002Obscurity,schneier2004nonsecurity} that is ``dangerous [for] the security of computer systems'' \cite{Bellovin2002Obscurity} and moreover ``contributes \dots to a false sense of security'' \cite{schneier2004nonsecurity}.

There are a number of rationales for the idea that obscurity increases the vulnerability of systems in practice. Perhaps the prominent is that ``public scrutiny is the only reliable way'' we know to improve security \cite{schneier2004nonsecurity}. Designing infallible systems is beyond the reach of the current and foreseeable engineering capabilities, and so ``failure is unavoidable and should be planned for'' \cite[p. 97]{ViegaMcGraw}.
In this context, ``[r]evealing and discussing vulnerabilities in hardware and software products allows the users to protect themselves, and encourages general protection and repair strategies'' \cite{Bellovin2002Obscurity}. Conversely, ``when security professionals withhold [vulnerability] information from the community, the broader community does not have an opportunity to find solutions'' \cite{Bellovin2002Obscurity} or to ``assess[\,] their own risks'' \cite{schneier2004nonsecurity}. ``[T]he result can be that the bad guys know about the problem long before most defenders do'' \cite{Bellovin2002Obscurity}. Another strand of reasoning suggests system owners may be more lax about secure design and mitigation of flaws under cover of secrecy~\cite{schneier2004nonsecurity}.

\iffullversion
\begin{quote}
    ``Before software bugs were routinely published, software companies denied their existence and wouldn't bother fixing them, believing in the security of secrecy. And because customers didn't know any better, they bought these systems, believing them to be secure.''\footnote{This quote is from 2004\ifecon~\cite{schneier2004nonsecurity}\fi; recent research suggests that companies denying the existence of bugs and not fixing them remains a practice that is not uncommon \cite{ParkThomas26}.} \ifecon\else\cite{schneier2004nonsecurity}\fi
\end{quote}
\fi

Yet the appeal of the ``security by obscurity'' argument has proven intuitive and enduring outside of the security community. The ideas that ``if the attacker doesn't know how it works, or where it is, they [are less likely to be able to] exploit it'' 
and that publishing vulnerability information simply ``give[s] attackers a roadmap'' to attack \cite{ristord-sbo} are intuitive --- and not wrong, either.
The essence of the established wisdom against security by obscurity is that the (perhaps less intuitive) security benefits of openness far outweigh these drawbacks \cite{schneier2004nonsecurity}.

Meanwhile, software and hardware vendors who ``want to keep vulnerability information secret'' \cite{schneier2004nonsecurity} --- not only to hide flaws, but also to shield intellectual property and other business interests --- have every incentive to take advantage of the intuitive ``security by obscurity'' narrative to argue for secrecy. More recently, AI companies have adopted similar narratives to justify model secrecy \cite{HallMP25}.

\medskip

But \emph{why} is there such a strong community consensus against ``security by obscurity''? Relevant sources point to the ``long history'' and experience of the field, sometimes with anecdotal evidence \cite{Bellovin2002Obscurity,Courtois09,Frei09}, and generally state the kinds of reasoning above without further justification. \iffullversion{What is more, d}\fi{etailed} sources on ``security by obscurity'' are relatively sparse (see \iffullversion{Hall et al.~}\fi\cite[\S2.1]{HallMP25} and \cref{sec:related}). The community consensus thus appears truly based more on collective experience than on written records\iffullversion\ or peer review\fi.

The history and experience of the field is venerable, and deserves deference. But arguments against ``security by obscurity'' would have more clout, especially outside the field, if supported by more theoretical or empirical foundations. Providing solid evidence to support our position is essential to reach a wider audience and counter the intuitive, recurring appeal of ``security by obscurity.'' It is, moreover, an urgent matter of public interest when our reliance on digital infrastructure is growing and cybersecurity threats are at an all-time high~\cite{ncsc-cybersec}.

\medskip

The rationales against ``security by obscurity'' are statements not of computer science, but rather, of security economics and human behavior. They are statements about how (rational) humans --- malicious hackers, ethical hackers, system owners, and users --- will behave in a realistic ecosystem where systems are vulnerable and need ongoing updates to improve their security. Specifically, they are statements about how those humans will behave with respect to systems that are ``obscured'' to deter scrutiny (e.g., by hiding how they work), versus systems that are ``open'' and invite scrutiny.

\begin{figure}[t]
  \centering
  
\begin{subfigure}[b]{0.32\textwidth}
    \centering
    \begin{tikzpicture}[
graphic node/.style={
        inner sep=0pt, rounded corners=2pt
      },
bipartite edge/.style={
        dotted,
        thick,
        gray!80
      }
    ]
    
\node[graphic node] (L1) at (0, 1) {\includegraphics[width=1.2cm]{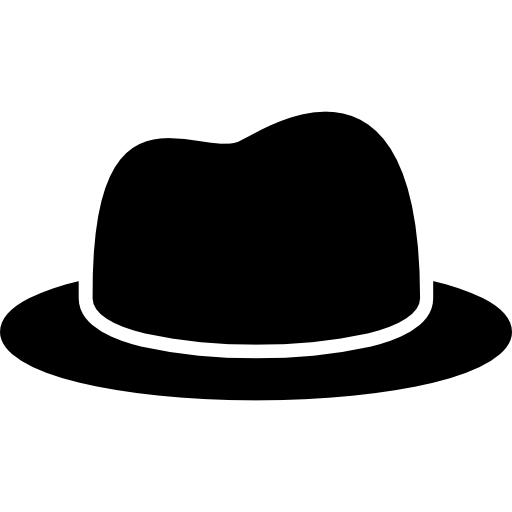}};
      \node[graphic node] (L2) at (0,-1) {\includegraphics[width=1.2cm]{figures/hat-b.png}};
    
\node[graphic node] (R1) at (2, 1) {\includegraphics[width=1.2cm]{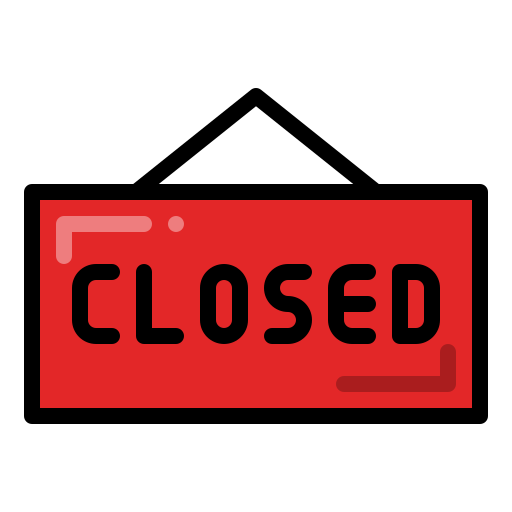}};
      \node[graphic node] (R2) at (2,-1) {\includegraphics[width=1.2cm]{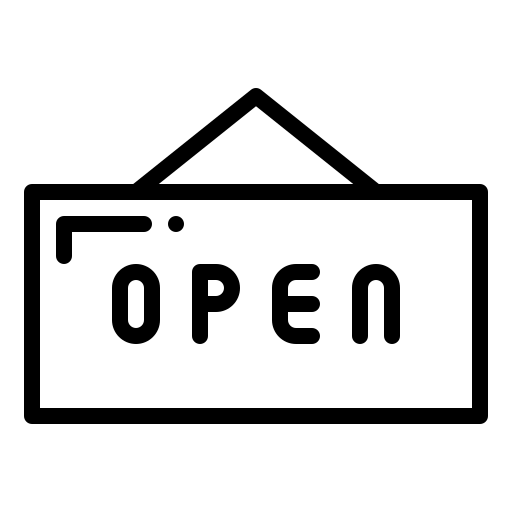}};
    
\draw[bipartite edge] (L1) -- (R1);
      \draw[bipartite edge] (L1) -- (R2);
      \draw[bipartite edge] (L2) -- (R1);
      \draw[bipartite edge] (L2) -- (R2);
    
    \end{tikzpicture}
    \caption{Simple Model (\S~\ref{sec:model-simple})}
    \label{subfig:simple}
  \end{subfigure}
  \hfill \begin{subfigure}[b]{0.32\textwidth}
    \centering
    \begin{tikzpicture}[
graphic node/.style={
        inner sep=0pt, rounded corners=2pt
      },
bipartite edge/.style={
        dotted,
        thick,
        gray!80
      }
    ]
    
\node[graphic node] (L1) at (0, 1) {\includegraphics[width=0.6cm]{figures/hat-b.png}};
      \node[graphic node] (L2) at (0, 0.33) {\includegraphics[width=0.6cm]{figures/hat-b.png}};
      \node[graphic node] (L3) at (0, -0.2) {\vdots};
      \node[graphic node] (L4) at (0,-1) {\includegraphics[width=0.6cm]{figures/hat-b.png}};
    
\node[graphic node] (R1) at (2, 1) {\includegraphics[width=1.2cm]{figures/closed.png}};
      \node[graphic node] (R2) at (2,-1) {\includegraphics[width=1.2cm]{figures/open.png}};
    
\draw[bipartite edge] (L1) -- (R1);
      \draw[bipartite edge] (L1) -- (R2);
      \draw[bipartite edge] (L2) -- (R1);
      \draw[bipartite edge] (L2) -- (R2);
      \draw[bipartite edge] (L4) -- (R1);
      \draw[bipartite edge] (L4) -- (R2);
    
    \end{tikzpicture}
    \caption{Generalization I (\S~\ref{sec:model-general})}
    \label{subfig:general1}
  \end{subfigure}
  \hfill \begin{subfigure}[b]{0.32\textwidth}
    \centering
    \begin{tikzpicture}[
graphic node/.style={
        inner sep=0pt, rounded corners=2pt
      },
bipartite edge/.style={
        dotted,
        thick,
        gray!80
      }
    ]
    
\node[graphic node] (L1) at (0, 1) {\includegraphics[width=0.6cm]{figures/hat-b.png}};
      \node[graphic node] (L2) at (0, 0.33) {\includegraphics[width=0.6cm]{figures/hat-b.png}};
      \node[graphic node] (L3) at (0, -0.2) {\vdots};
      \node[graphic node] (L4) at (0,-1) {\includegraphics[width=0.6cm]{figures/hat-b.png}};
    
\node[graphic node] (R1) at (2.2, 1.1) {\includegraphics[width=1cm]{figures/closed.png}};
      \node[graphic node] (R2) at (2.2, 0.33) {\includegraphics[width=1cm]{figures/open.png}};
      \node[graphic node] (R3) at (2.2, -0.25) {\vdots};
      \node[graphic node] (R4) at (2.2,-1) {\includegraphics[width=1cm]{figures/closed.png}};
    
\draw[bipartite edge] (L1) -- (R1);
      \draw[bipartite edge] (L1) -- (R2);
      \draw[bipartite edge] (L1) -- (R4);
      \draw[bipartite edge] (L2) -- (R1);
      \draw[bipartite edge] (L2) -- (R2);
      \draw[bipartite edge] (L2) -- (R4);
      \draw[bipartite edge] (L4) -- (R1);
      \draw[bipartite edge] (L4) -- (R2);
      \draw[bipartite edge] (L4) -- (R4);
    
    \end{tikzpicture}
    \caption{Generalization II (\S~\ref{sec:model-general-manysystems})}
    \label{subfig:general2}
  \end{subfigure}

  \caption{Overview of Our Models}
  \label{fig:models}
\end{figure} 
\paragraph{Our contributions}
We present the first game-theoretic models and analysis of ``security by obscurity,'' examining one prominent rationale against obscurity: namely, that obscuring systems can result in ``the bad guys know[ing] about the problem[s in that system] long before most defenders do'' \cite{Bellovin2002Obscurity}. 

We first introduce a simple model with two players and two systems (\cref{subfig:simple}). Players may choose to examine a system and look for vulnerabilities, and they come in two types, which we call \emph{good} and \emph{bad}, corresponding to the ``defenders'' and the ``bad guys'' in the preceding quote. 
We assume good players disclose vulnerabilities and seek to have them fixed, whereas bad players exploit vulnerabilities (or sell them for exploitation by others), and model the utility gained by good players on finding a vulnerability as generally lower than the utility gained by bad players (see \cref{sec:model-overview}).
Systems also come in two types, \emph{open} and \emph{closed}, where closed systems are modeled as more costly or difficult to examine than open ones. The probability of a hacker first finding a vulnerability in given system (which we call ``winning'') is modeled by a probabilistic skill-based contest rule (\eqref{p(s)} in \cref{sec:model-overview}).
We then present generalized models with many hackers and many systems (\cref{subfig:general1,subfig:general2}).
This type of model is known as a \emph{strategic entry game}.

In all our models, we prove that in equilibrium under realistic conditions, closed systems are more likely than open systems to have their vulnerabilities first discovered by bad hackers rather than good hackers. As good hackers will disclose vulnerabilities and seek to have them fixed, whereas bad hackers will exploit vulnerabilities (or sell them for exploitation by others), this means that closed systems are less secure than open systems in that they are more likely to suffer successful attacks based on a vulnerability before it is fixed.

Our result can be understood as an adverse-selection effect. A closed system is attractive only to hackers whose expected benefit from targeting that system is high enough to justify the higher cost of examination. If bad hats value finding vulnerabilities more (deterministically or stochastically), they are more willing to pay that cost. Therefore, at any given skill level, bad hackers are more likely to examine a closed system, while good hackers more likely prefer an open system.

Two equilibrium forces are central to our analysis. First, {\it congestion} makes a system less attractive as more hats enter it. Second, {\it selection} arises because systems with higher entry costs are chosen only by hackers whose expected benefit is sufficiently high. The interaction between congestion and selection determines both the number and composition of hackers who choose to examine a given system and, consequently, the probability that each type wins each system.

\paragraph{Roadmap}
Next, \cref{sec:related} briefly covers related work. Section~\ref{sec:model-overview} introduces our models and results at a high level, and justifies modeling choices. Section~\ref{sec:model-prelim} establishes notation and preliminaries, \cref{sec:model-simple} details our simple model and analysis, and \cref{sec:model-general,sec:model-general-manysystems} present our generalized models and analysis. Finally, \cref{sec:discussion} offers discussion.

\bigskip

\section{Related Work}\label{sec:related}

Before modern cryptography and systems security, even critical military communications often relied on security by obscurity~\cite{Tohe2022codetalkers,KatzLindell}. We refer to \iffullversion{Kahn~}\fi\cite{KahnCodebreakers} for a history of security by obscurity in the context of codemaking and codebreaking. 

Modern security literature tends to frame ``security by obscurity'' as a dangerous fallacy that is unlikely to improve security at best, and likely to harm security at worst, citing a range of reasons already discussed in \cref{sec:intro}~\cite{Bellovin2002Obscurity,Shipman2019opencode,Mercuri2003obscurity,schneier2004nonsecurity,diffie2003risky}; however, academic treatments are scarce. To our knowledge, there are no rigorous theoretical or empirical analyses to date of the rationales against ``security by obscurity.'' On the empirical side, this may relate to the practical difficulties of obtaining a realistic data set and of analyzing what would have happened in the alternative (i.e., analyzing what would have happened if a system that was open/closed in reality had instead been closed/open).

Security economics has long established that ``security failure is caused by bad
incentives at least as often as by bad design'' \cite{AM07,ross_econ}. The field of security economics considers how incentives influence security outcomes, and our work lies squarely within its domain. However, we are unaware of prior security economics work that has rigorously analyzed the rationales against ``security by obscurity.''

Prior work has documented the motivations of ethical hackers \cite{hackerone-hacker,ParkThomas26} --- that they often see research as a public service, and feel a duty to report findings even at considerable risk to themselves, and even when they could profit financially much more by exploiting or selling their findings instead. 
\iffullversion
    We base aspects of our model on the findings of \ifecon{Park and Thomas~}\fi\cite{ParkThomas26}, as further discussed in \cref{sec:model-overview}.
\fi

\section{Overview of Our Models}\label{sec:model-overview}

Next, we introduce our models and justify our modeling choices at a high level. 

\paragraph{The Game in a Nutshell}
Our models consider a set of players and a set of systems which contain vulnerabilities. Each player may choose to ``target'' any of the systems and look for vulnerabilities in them, or to target no system. We say that the player\blu{(s)} who first find\blu{\sout{s}} a reportable or exploitable vulnerability in a given system ``wins'' that system. Winning yields a reward, but vulnerability hunting incurs some cost, so which system to target is a strategic decision\iffullversion{ for players}\fi.

Our main models assume a single winner per system and players that target one system, for simplicity, but they extend naturally to multiple winners and players targeting multiple systems (see \cref{rmk:multiple-winners,rmk:multiple-targets}). Searching for vulnerabilities in systems over an extended period can be modeled by a sequence of games like ours, with an updated game starting once a system is won.

\paragraph{The Players (``Hats'')}
Players come in two \emph{ethical types}, which we call good ($\G$) and bad ($\B$). These correspond to ``white-hat'' and ``black-hat'' hackers, respectively: white hats ``are cybersecurity professionals who use their skills to identify and fix security vulnerabilities, helping to protect systems and data from attacks,''
whereas black hats maliciously exploit security vulnerabilities for ``personal gain, such as stealing sensitive data, disrupting systems, or committing fraud''\ifecon.\cite{hackerone-hats}\else~\cite{hackerone-hats}.\fi\footnote{These are not the only two traditional categories of hackers~\cite{hackerone-hats}; however, modeling these two types suffices to demonstrate our point.} In a nod to this terminology, we often refer to players as ``hats.'' 

Selling or exploiting vulnerabilities can be quite profitable: for example, important vulnerabilities can be worth ``life-changing money'' \cite{ParkThomas26}, and ransomware and data theft can extort millions of dollars \cite{data-theft}. In contrast, white-hat hackers find and report vulnerabilities so that they can be fixed tend to be motivated by the desire to improve the security of systems, and may also receive smaller-scale personal benefit in the form of professional reputation or bug bounties \cite{ParkThomas26}. The professional and financial gain of operating as a white hat is generally far lower than the potential profit from operating as a black hat\ifecon.\cite{ParkThomas26}\else~\cite{ParkThomas26}.\fi\footnote{As one anonymous white-hat hacker put it,
    ``I’m not at a place where I need [the money] at this point, so I am definitely more likely to disclose \dots but I can’t say that that would’ve been the case \dots 10 years ago, 15 years ago, I would’ve taken every dollar I could get \dots  blood money or not. \dots
	It’s easy when I don’t worry about my mortgage [or my] bills, but if I was worrying about any of that \dots there’s no way that I’d [report the finding] as opposed to [selling it] \dots'' \cite{ParkThomas26}.
    \label{fn:bloodmoney}
}

The simplest perspective on our models is that hats are either good or bad \emph{intrinsically}: their ethical type determines whether they report or exploit a vulnerability upon discovery, and their utility 
\iffullversion
    from winning 
\fi
differs between the two types systematically. Personal differences within each ethical type can then be captured by a motivation parameter $z_i\in\Re_+$ assigned to each player.\footnote{\iffullversion Prior research~\cite{ParkThomas26} documents some of the ways that personal ethical decision-making differs between individual ethical hackers.\ \fi Compare, for example, \cref{fn:bloodmoney} with an ethical hacker who reported declining to sell their findings to a mysterious buyer who approached them: ``I did think about it though, right, like holy shit, that’s life-changing money. `I could move my entire family out of the `hood' money, right --- so it was an ethical debate for me.'' \cite{ParkThomas26}.}
An alternative approach could have all hats decide whether to fix or exploit vulnerabilities upon finding them, based on the utilities they anticipate from the two options. In this case, the financial profit from exploitation for those who report can be thought to be offset by a large negative utility arising from a personal ethical preference, which may vary by individual. Such individual variance can, also here, be captured by a motivation parameter $z_i$. In reality, of course, all options --- from maliciously exploiting the vulnerability to ethically reporting it --- are generally available.

Our models take the simpler former approach. However, as detailed in \cref{app:alt-model}, the latter modeling approach --- that all options are available to all players --- is reducible to our models, essentially by labeling those who rationally choose to report or exploit as ``good'' and ``bad'' respectively.

Regardless of their type $\eps\in\BG$ and motivation $z$, each player also has a \emph{skill level} $s$ in a closed interval $S\subset\Re_+$ which affects their chances of winning according to a standard \emph{weighted contest-success probability}, where player $i\in[n]$ has the following probability of winning in an $n$-player contest: 
\begin{equation}
p_i(s_1,\cdots,s_n)\ :=\ \frac{w(s_i)}{\sum_{j=1}^nw(s_j)}\ \in\ (0,1]\ ,
\label{p(s)}
\end{equation}
where $s_j$ is the skill level of player $j\in[n]$ and  $w:S\to(0,\infty)$ is a continuous and strictly increasing weight function.

In sum, the personal characteristics of each player are $(\eps,s,z)$. We model these as drawn independently from publicly known distributions for every player. In particular, we model the distributions of skill as the same across good and bad hats, in order to isolate the impact of utility differences between ethical types.

\paragraph{The Systems} 
We model two types of systems: open ($\O$) and closed ($\C$). 
The idea of openness here encompasses system transparency and documentation, as well as openness to scrutiny more broadly (e.g., 
\iffullversion
    pledging not to antagonize or sue independent good-faith researchers, 
    a risk that is known to deter research, 
    or encouraging and offering payments to researchers for reporting their findings \cite{ParkThomas26}).
\fi
 As such, open systems can be examined at a lower cost than closed systems. Our models assume, for simplicity, that examining open systems costs zero\footnote{This is true, for example, when players discover vulnerabilities by accident; prior work has documented that this is not uncommon~\cite{ParkThomas26}. Our results hold as long as examining open systems costs zero for some fraction, however small, of good hats. Perhaps surprisingly, our results do not immediately generalize when open systems have small non-zero entry costs for every player; \cref{rmk:open-sys-entry-costs} (\cref{app:thm2-remarks}) discusses why we believe, in practice, open systems are more secure in this case nonetheless.} and examining closed systems costs more than zero.

\qparagraph{Which Type of System Is More Secure} We assess which type of system is more secure by analyzing the probabilities that it will be won by a white hat or a black hat. We consider a system more secure than another if it has a higher probability of being won by a white hat (who will not exploit it, and will enable fixing the security issue), and a lower probability of being won by a black hat (who will maliciously exploit it or enable others to do so). This conception of security parallels the received wisdom that one of the problems with security by obscurity is ``that the bad guys know about the problem long before most defenders do'' \cite{Bellovin2002Obscurity}. Our analysis shows that open systems are ``more secure'' than closed systems in this sense, consistently across our models.

\paragraph{Simple Model: 2 Hats, 2 Systems} Our simple model (\cref{sec:model-simple}, depicted in \cref{subfig:simple}) considers just two hats and two systems: one open system and one closed system that are identical in all respects except for being open vs. closed (i.e., cost to examine). Bad hats derive greater utility from winning than good hats do. The simple model omits the motivation parameter~$z$.

\paragraph{Generalizing: Many Hats} Our first generalized model (\cref{sec:model-general}, depicted in \cref{subfig:general1}) considers hats with the full set $(\eps,s,z)$ of characteristics, and two systems, one open and one closed, which have different intrinsic values --- representing the idea that some systems are more valuable targets than others. Bad hats no longer deterministically derive greater utility from winning than good hats; instead, the utility of winning for any given hat follows a distribution that depends on the system's intrinsic value and bad hats \emph{likely} derive greater utility from winning\iffullversion\ than good hats\fi.

\paragraph{Generalizing: Many Systems}
Our second generalized model builds on the first by considering multiple systems which may be open or closed, with an intrinsic value assigned to each system, and costs of examination varying across systems.

\qparagraph{Why Is Equilibrium the Condition We Care About} 
A Nash \e ``captures a steady state of the play of a strategic game in which each player holds the correct expectation about the other players’ behavior and acts rationally'' \cite[p. 14]{OsborneRubinstein94}. Definitionally, in an \e, no individual player would benefit from disrupting it by behaving differently. 
Nash equilibrium is the standard solution concept studied in game theory.
More pragmatically, an important line of economics literature has suggested that equilibria may ``arise[\,] as the long-run outcome of a process in which less than fully rational players grope for optimality over time'' \cite{FudenbergLevine98,MilgromRoberts90}.
Experimental evidence moreover supports equilibrium as a possible steady state of experienced or repeated interaction \cite{CooperDejongForsytheRoss90,VanhuyckBattalioBeil90}.

\paragraph{Overview of Our Results}
In essence, we examine whether either of the two opposing views on ``security by obscurity'' is justifiable as stable behavior of rational hackers (i.e., an \e behavior of a game) in a plausible model that captures real-world incentives (as documented, e.g., by \cite{ParkThomas26}). 
We believe such a model should reflect that: (i) bad hats derive higher rewards than good hats, stochastically; (ii) closed systems are not impenetrable; and (iii) the number of potential hackers is substantial. We show that the folklore is justified in every \e: i.e., any open system has a strictly higher (lower, resp.) \emph{unconditional} probability to be won by a $\G$-hat ($\B$-hat, resp.) than any closed system.

As a preliminary step, we show in a simple model (2-hat, 2-system) that the ordering prevails conditional on either system having a winner, essentially because the closed system is more likely targeted by hackers expecting higher rewards (who tend to be $\B$-hats). We then show that this result extends to the unconditional likelihood of winner type when there are sufficiently many hats, in both generalized models\iffullversion.\ 
    This last step is not obvious, especially when the fraction $\beta$ of $\B$-hats is small.
\fi

Our results are meaningful in practice when they do not require an unrealistically large number of hats. One may wonder how large the number of hats must be, especially when $\beta$ is relatively small. We argue that our results are realistic because the probability of a system having no winner vanishes rapidly as the expected number of hats who target that system increases.
For example, if $I$ hats target a system with the same probability so that the expected number targeting that system is $n$, the probability of no winner (i.e., of no one targeting the system) is $(1-n/I)^I$\iffullversion, which decreases in $n$ and increases in $I$\fi: thus, already for $n=5$ for example, $(1-n/I)^I|_{n=5} \le \lim_{I\to\infty}(1-n/I)^I|_{n=5}= e^{-5}\approx 0.0067$.
\cref{app:simulation}'s approximations show the number of hats needed is realistic
for sensible parameters.

\paragraph{Discussion}
We consider \cref{thm:main2} our main result; however, we note that our simple model is incomparable. The simple model illustrates the essence of our result with minimal analytic complexity. Our analysis of the generalized models show that our result still holds for larger numbers of systems and hats, when taking into account more realistic variation in their properties.

Our models are necessarily stylized, not capturing all realistic factors at play. We view their simplicity as a strength, given our aim to demonstrate a core dynamic of strategic behavior based on minimal modeling assumptions. We discuss our models' generalizability in several dimensions, in \cref{app:thm2-remarks}.

\section{Notation and Game-Theory Preliminaries}\label{sec:model-prelim}

\paragraph{Notation} $x\sim D$ denotes a random variable $x$ has distribution $D$, 
``i.i.d.'' denotes ``independent and identically distributed.'' ``CDF'' denotes ``cumulative distribution function.'' $\Delta(X)$ denotes the set of all probability distributions over a set $X$, a probability distribution is \emph{atomless} if it assigns 0 to every singleton $\{x\}\subset X$, and $[n]$ denotes the set $\{1,\dots,n\}$ for 
\iffullversion
    a natural number 
\fi
$n\in\mathbb{N}$. 
\iffullversion
    $\text{supp}$ denotes 
    the support of a function 
    and $\inf$ denotes infimum.
\fi

\smallskip

A \emph{game} specifies a set of players indexed by $i\in [I]$, an available \emph{action set} $A_i$ for each player, an \emph{outcome function} $\zeta:A_1\times\dots\times A_I\to \Delta(\Omega)$ where $\Omega$ is an outcome space,  and a \emph{utility function} $u_i:\Omega\to\Re$ for each player $i$ that assigns a payoff or \emph{utility} to every possible outcome.
The outcome experienced by each player depends not only on their own action, but also on the actions of others. 

We consider {\it Bayesian games}, in which each player possesses private information $\theta_i\in\Theta_i$ which we call the player's \emph{type}, drawn from a known distribution, and players' outcome and utility functions are contingent on these types, i.e., 
\[
\zeta:\Theta_1\times\cdots\times\Theta_I\times A_1\times\dots\times A_I\to \Delta(\Omega)
\quad\mbox{and}\quad u_i:\Theta_i\times\Omega\to\Re\ .
\]
Each player $i$, knowing their own type $\theta_i$, chooses an action $a_i\in A_i$ without observing others' types or action choices. A \emph{strategy} of player $i$ is a mapping, $\a_i:\Theta_i\ \rightarrow\ \Delta(A_i)$, from $i$'s type to a probability distribution over actions.
The description of the game is common knowledge.\footnote{``[Some information] is common knowledge if players know [it], know that the other players know [it], and so on {\it ad infinitum}.'' \cite[p. 541]{FudenbergTirole91}} 

In our model, each player $i$ has a private \emph{type} $\theta_i=(\eps_i,s_i,z_i)\in\Theta$, where $\eps_i\in\BG$ is hat $i$'s \nature (bad or good), $s_i\in S$ is a skill parameter, and $z_i\in\Re_+$ is a motivation parameter. 
The action set $A_i$ comprises a set of systems of which any hat may target one, and the option of not targeting any system, which we denote by $\phi$;
an outcome is specified by which system each hat targeted and a winner (if any) of each system. 
Hence, $\Theta_i=\Theta$ and $A_i=A$ for all $i$; from now, we just use $\Theta$ and $A$.

An \emph{action profile} $a=(a_i)_{i\in[I]} \in A^{I}$ specifies an action for each player. Contingent on a \emph{type profile} $\theta=(\theta_i)_{i\in [I]} \in \Theta^{I}$
which specifies a type for each player, an action profile $a$ determines the expected payoff of each player $i$ as 
\[
\E \big[u_i(\theta_i, \zeta(\theta,a))\big] = \int_\Omega u_i(\theta_i,\omega)\ d\zeta(\theta,a)\ ,
\]
where the expectation is taken over outcomes generated by $\zeta(\theta,a)\in\Delta(\Omega)$.\footnote{This assumes the requisite measurability of $\zeta$ (for the integral). Throughout, we assume the requisite measurability conditions whenever integration is invoked.} Abusing notation, we use $u_i(\theta,a):=\E[u_i(\theta_i, \zeta(\theta,a))]$ to denote player $i$'s expected payoff for an action profile $a$ and type profile $\theta$.

A {\it strategy profile} $\a=(\a_i)_{i\in [I]}$, specifies a strategy $\a_i$ for each player $i$. 
We may write $\a=(\a_i,\a_{-i})$ where $\a_{-i}=(\a_j)_{j\in [I]\setminus\{i\}}$, and similarly for $a,\theta$.

A strategy profile $\a$ is a {\it Bayesian Nash equilibrium}\footnote{A Bayesian Nash equilibrium is a Nash equilibrium in a Bayesian game.} if no player of any type can improve its expected payoff by changing their strategy while all other players keep their strategies unchanged,\footnote{Game theory assumes that each player's preference over possible outcomes has an expected-utility  representation \cite{vonNeumannMorgenstern44}, hence each player's objective is to maximize their expected payoff given the strategies of others. } that is, if for all $\theta_i\in\Theta$, $\a'_i\in \Delta(A)$ and $i\in [I]$,
\[
\E_{\theta_{-i}}\big[
u_i(\theta,\a_i(\theta_i),\a_{-i}(\theta_{-i}))\big]
\ge \E_{\theta_{-i}}\big[u_i(\theta,\a_i',\a_{-i}(\theta_{-i}))\big]\ ,
\]
where the expectation is taken over the private types $\theta_{-i}$ of every player other than $i$. For brevity, we write just ``equilibrium'' in the rest of the paper.

\section{Simple Model}\label{sec:model-simple}

We start with the simplest possible model that conveys our core result (\cref{subfig:simple}): there are two systems, one open and one closed, and two hats $i=1,2$. Each hat is bad ($\eps_i=\B$) with probability $\beta\in(0,1)$ and good ($\eps_i=\G$) with probability $1-\beta$, and has a private i.i.d. skill level $s_i \sim F$ where $F$ is an atomless CDF with full support on a compact (i.e., closed and bounded) interval $S$ normalized to $[0,1]$. Each hat decides whether to target the open system ($\O$) at zero cost, target the closed system ($\C$) for a cost $k>0$, or not target either, i.e., $A=\{\O,\C,\phi\}$.

A hat $i$ targeting a system $\sys\in\{\O,\C\}$ wins it with a standard weighted contest-success probability as introduced in \eqref{p(s)}. From winning either system, hat $i$ obtains a reward of $V_{\eps_i}$ depending on their ethical type $\eps_i\in\BG$ where $V_\B>V_\G>0$. Hence, the expected (net) payoff of targeting $\sys$ is $V_{\eps_i}$ times the probability of winning it if $\sys=\O$, and that minus $k$ if $\sys=\C$. 
The simple model omits the motivation parameter $z_i$,
hence hat $i$'s type is $\theta_i=(\eps_i,s_i)\in\BG\times S$.

\begin{lemma}\label{prop:simple}\label{lem:simple}
If $V_\B\ge 2k$ in the simple model described above, 
the likelihood that the closed system $\C$ is won by a bad hat (conditional on $\C$ having a winner) is strictly greater than the likelihood that the open system $\O$ is won by a bad hat (conditional on $\O$ having a winner).
\end{lemma}

\begin{proof}
Consider an arbitrary strategy profile $\a=(\a_1,\a_2)$ for the two hats, where
\[
 \a_i: \BG\times S \to\ \Delta\big(\{\O, \C,\phi\}\big)\ .
\]
From hat $i$'s perspective, the rival hat, which we denote by $-i$, targets system $\sys\in\{\O,\C\}$ with probability
\begin{equation}
q^\sys_{-i}(s_{-i}|\a_{-i})\ :=\ \beta\cdot\a_{-i}(\B,s_{-i})(\sys)\,+\, (1-\beta)\cdot\a_{-i}(\G,s_{-i})(\sys)\ , \label{qiT}
\end{equation}
depending on their skill level $s_{-i}$, where $\a_{-i}(\eps_{-i},s_{-i})(\sys)$ is the probability that $\a_{-i}$ assigns to targeting $\sys$ when their type $\theta_{-i}$ is $(\eps_{-i},s_{-i})$. 
Upon targeting system $\sys$, therefore, hat $i$ wins $\sys$ with probability 
\begin{equation}
p_i^\sys(s_i|\a) := 
\int_{s_{-i}} \left[p_i(s_1,s_2)\cdot q_{-i}^\sys(s_{-i}|\a_{-i})+
\left(1-q_{-i}^\sys(s_{-i}|\a_{-i})\right)\right]\,dF\ , \label{piT}
\end{equation}
because they win with probability $p_i(s_1,s_2)$ if hat $-i$ targets the same system and $1$ otherwise (as defined in \eqref{p(s)}).
Note that this yields positive winning probability conditioned on targeting a system.

Thus, each hat $i$ of type $\theta_i=(\eps_i,s_i)$ obtains an expected payoff $V_{\eps_i}\cdot p^\O_i(s_i|\a)$ from targeting system $\O$ and $V_{\eps_i}\cdot p^\C_i(s_i|\a)-k$ from targeting $\C$. Since the former payoff is positive (since $V_{\eps_i}>0$ by assumption) and the payoff from targeting neither system is zero, it is optimal for hat $i$ of type $(\eps_i,s_i)$ to act as follows:
$$
\begin{array}{rll}
\multirow{2}{*}{\hbox{if~~$V_{\eps_i}\cdot\big(p_i^\C(s_i|\a)-p_i^\O(s_i|\a)\big)$} ~~\Big\{} & \quad \ge k, ~~ \mbox{target}~\C\\
& \quad \le k, ~~ \mbox{target}~\O\ .
\end{array}
$$

Therefore, $\a=(\a_1,\a_2)$ is an \e if and only if
\begin{equation}
\text{supp}(\a_i(\eps_i,s_i))\subseteq\ \begin{cases} \{\C\} &\mbox{if}\ \ V_{\eps_i}\cdot\big(p^\C_i(s_i|\a)-p^\O_i(s_i|\a)\big)> k\\
\{\O,\C\} &\mbox{if}\ \ V_{\eps_i}\cdot\big(p^\C_i(s_i|\a)-p^\O_i(s_i|\a)\big)= k\\
\{\O\} &\mbox{if}\ \ V_{\eps_i}\cdot\big(p^\C_i(s_i|\a)-p^\O_i(s_i|\a)\big)< k\ . \end{cases}  \label{A*i}
\end{equation}

Because $V_\B>V_\G$, bad hats are more willing than good hats to pay the cost of targeting $\C$. Hence, for a given skill level $s_i$, if a $\G$-hat would target $\C$, so would a $\B$-hat; and if a $\B$-hat would target $\O$, so would a $\G$-hat. This divides skill levels $s_i$ into three conditions depending on the \e behavior of hats, as below:

\begin{enumerate}[label=(\roman*)]
    \item $p^\C_i(s_i|\a)-p^\O_i(s_i|\a)\ge k/V_\G\ \Longrightarrow\ $ target $\C$ regardless of \nature $\eps_i$ (unless the inequality binds and $\eps_i=\G$, in which case target either system), \label{itm:condition-ge}
    \item $p^\C_i(s_i|\a)-p^\O_i(s_i|\a)\le k/V_\B\ \Longrightarrow\ $ target $\O$ regardless of \nature $\eps_i$ (unless the inequality binds and $\eps_i=\B$, in which case target either system),~and \label{itm:condition-le}
    \item $k/V_\B<p^\C_i(s_i|\a)-p^\O_i(s_i|\a)< k/V_\G\ \Longrightarrow\ $ target $\C,\O$ if $\eps_i=\B,\G$ respectively. \label{itm:condition-between}
\end{enumerate}

\noindent These three conditions are mutually exclusive, since $k/V_\B < k/V_\G$.

\begin{observation}\label{prop:bad-more-likely1}
    Conditional on having a skill level $s_i$ satisfying \ref{itm:condition-ge} or \ref{itm:condition-le}, a hat is a $\B$-hat with a probability weakly higher (lower, resp.) than $\beta$ if they target $\C$ ($\O$, resp.), and strictly so for $s_i$ satisfying \ref{itm:condition-between}.
\end{observation}

Since winning probabilities depend on skill (see \eqref{p(s)}) but not on a hat's \nature $\eps_i$,
\cref{prop:bad-more-likely1} 
moreover holds for the winning hat, and is strict if \ref{itm:condition-between} holds with positive probability (i.e., for a positive measure of skill levels).

Let us define $\tilde{p}_i:= p^\C_i(s_i|\a)-p^\O_i(s_i|\a)$, which is continuous in $s_i$. Since \ref{itm:condition-ge} requires $\tilde{p}_i\ge k/V_\G$ while \ref{itm:condition-le} requires $\tilde{p}_i\le k/V_\B<k/V_\G$, if conditions \ref{itm:condition-ge} and \ref{itm:condition-le} each prevail at some skill level, the
intermediate value theorem implies that there is a nonempty interval of skill levels outside of \ref{itm:condition-ge} and \ref{itm:condition-le}. Hence, condition \ref{itm:condition-between} must prevail throughout that interval, and thus with positive probability under our full-support assumption on the skill
distribution $F$. 

It remains to consider the possibilities that condition \ref{itm:condition-ge} covers the entire skill space $S$, or \ref{itm:condition-le} does.
In the former case, since both hats target $\C$ with probability $1$, either hat can profitably deviate to $\O$, where it would then win for sure and avoid the cost $k>0$. Hence,  condition \ref{itm:condition-ge} cannot cover all of $S$ in \e.

If condition \ref{itm:condition-le} covers $S$, both hats target $\O$ with probability $1$, obtaining an expected payoff of $V_{\eps_i}\cdot \E _{s_{-i}}[p_i(s_1,s_2)]$. For equilibrium to hold, the payoff $V_{\eps_i}-k$ that they would get by targeting $\C$ alone instead should be no higher for all types. This condition, namely,
\iffullversion
    \[
    V_{\eps_i}-k\ \leq\ V_{\eps_i}\cdot \E _{s_{-i}}[p_i(s_1,s_2)]\qquad \forall (\eps_i,s_i)\in\BG\times S\ ,
    \]
\fi
is possible only if $V_\B\!<\!2k$ because
\iffullversion
    $$
    \inf \E _{s_{-i}}[p_i(s_1,s_2)] \quad\le\quad \E _{s_{-i}}[p_i(s_1,s_2)]\big|_{s_i=\min S} \quad<\quad 1/2\ .$$
\fi
\cref{prop:simple} follows. \qqed
\end{proof}

\iffullversion
    In sum, in our simple model, the open system is more secure in the sense of \cref{prop:simple}, except when $k$ is large enough ($V_\B<2k$) that the closed system is not worth pursuing even when all hats compete in the open system.
\fi

\section{Generalized Model I: Many Hats}\label{sec:model-general}

\label{sec:model-general-hats}

There are again two systems, an open system $\O$ and a closed system $\C$, but now many hats, the number $I$ of which is uncertain. Following Myerson's seminal work \cite{myerson98} which first defined games with a large, uncertain number of players, we assume that $I$ is a Poisson random variable with mean $\lambda$, i.e., $I\sim \operatorname{Poisson}(\lambda)$.\footnote{This means $I$ equals any nonnegative integer $n$ with probability $e^{-\lambda}\cdot \lambda^n/n!$.\label{fn:poisson-def}} We, like Myerson, work with the Poisson distribution because of its convenient properties, such as \emph{environmental equivalence} (see footnote \ref{fn:env eqv}); but our results hold more broadly, as discussed in \cref{rmk:poisson}.

Each hat $i$ is either bad or good with probability $\beta$ and $1-\beta$, respectively, and has a skill level $s_i \sim F$ as before. A hat $i$ targeting a system $\sys\in\{\O,\C\}$ wins with the weighted contest-success probability (\ref{p(s)})\iffullversion\ depending on the number of all hats targeting the same system and their skill levels\fi.
The reward from winning system $\sys$ is $z_i\cdot V_\sys$ where $V_\sys$ is a commonly known, intrinsic value of the system $\sys$ (which may differ across systems) and $z_i$ is hat $i$'s private motivation parameter drawn from an atomless distribution with CDF $H_{\eps_i}$, associated density function $h_{\eps_i}$, and full support $Z:=\Re_+$,\footnote{Our results hold when $Z$ is bounded so long as $\max Z$ is sufficiently large; see \cref{app:thm2-remarks}.} depending on \nature $\eps_i\in\BG$.  $H_\B$ first-order stochastically dominates $H_\G$ strictly, meaning:
\begin{equation}
H_\B(z)\ <\ H_\G(z)\qquad\forall z>0\ . \label{fosd}\end{equation}
Thus, individual hats value each system differently. However, 
\iffullversion
    \eqref{fosd}
\fi
means that $\B$-hats value systems more than $\G$-hats, stochastically. This generalizes our simpler model, in which hats value both systems identically and $\B$-hats value them more than $\G$-hats by a fixed amount, to a wider range of systems and hats' preferences.

Since hat $i$'s type is now $\theta_i=(\eps_i,s_i,z_i)$, a {\it strategy} of hat $i$ is a function 
\[
\a_i: \BG\times S\times Z \to\ \Delta\big(\{\O, \C,\phi\}\big)\ .
\]

Symmetric equilibrium is the accepted solution concept for Poisson games, because from any player's perspective not knowing who else is playing, all other players are indistinguishable \cite{myerson98}.\iffullversion\footnote{``It cannot be commonly perceived that two different individuals of the same type would behave differently because \dots [they] have no commonly known attributes by which others can distinguish them. \dots 
In effect, population uncertainty forces us to treat players symmetrically in our game-theoretic analysis.'' \ifecon\else---Myerson\fi \cite{myerson98}.}\fi\ 
Accordingly, in our generalized models, by ``\e,'' we mean a symmetric \e (i.e., in which $\a_i=\a$ for all $i$).

\begin{theorem}\label{thm:main} 
The following holds in every \e in the two-system Poisson game just described. 
\begin{enumerate}[label=(\alph*)]
\item The open system $\O$ is strictly more likely to be won by a $\G$-hat than the closed system $\C$ is, conditional on the system having a winner. \label{itm:thm:conditional-on-winner}
\item The unconditional probability that a system is won by a $\G$-hat is strictly higher for $\O$ than for $\C$, if $\lambda$ is large enough. \label{itm:thm:lambda}
\item The unconditional probability that a system is won by a $\B$-hat is strictly higher for $\C$ than for $\O$, if $\lambda$ is large enough and $h_\G$ has a thinner upper~tail: \label{itm:thm:lambda-and-tail}
\begin{equation}
\frac{h_\G(z)}{h_\B(z)} \le \chi\ \text{ for some $\chi<1$, for all $z$ above some threshold.} \label{uppertail2}
\end{equation}
\end{enumerate}
\end{theorem}

The last condition \eqref{uppertail2} means that bad hats are more likely than good hats to have motivation parameters above a threshold, reflecting the idea that bad hats have access to larger financial rewards (see \cref{sec:model-overview}).

Next, we give a proof sketch detailing the main logical steps that lead to our results, not relying on the specific distribution of $I$ as Poisson. \cref{app:proof2} provides a full proof assuming a Poisson distribution; however, our results hold for other distributions so long as they do not invalidate the high-level steps sketched below.

\begin{proof}[sketch]
In \e, the expected winning probability of hat $i$ of skill $s_i$ targeting system $\sys\in\{\O,\C\}$ is  the expectation of $ p_i(s_i,s_{-i}|_\sys)$ where $s_{-i}|_\sys$ is the skill profile of other hats who also target $\sys$, conditional on all other hats behaving according to the \e strategy $\a$. Let $P_\sys\!(s_i|\a)$ denote this expected winning probability; we derive its expression in \cref{app:Psys}.
Then, hat $i$ of type $\theta_i=(\eps_i,s_i,z_i)$ expects a net payoff of 
\[
z_i\cdot V_\O\cdot P_\O(s_i|\a) \qquad\text{and}\qquad
z_i\cdot V_\C\cdot P_\C(s_i|\a) - k\ ,
\]
respectively, from targeting systems $\O$ and $\C$.

Since $z_i\cdot V_\O\cdot P_\O(s_i|\a)>0$ and the payoff from not targeting any system is zero, analogously to before (see \eqref{A*i}), $\a$ is an \e if
\begin{equation}
\text{supp}\big(\a(\eps_i,s_i,z_i)\big)\ \subseteq\ \begin{cases} \{\C\} &\mbox{if}\ \ z_i\cdot\big(V_\C\cdot P_\C(s_i|\a)- V_\O\cdot P_\O(s_i|\a)\big)> k\\
							\{\O,\C\} &\mbox{if}\ \ z_i\cdot\big(V_\C\cdot P_\C(s_i|\a)- V_\O\cdot P_\O(s_i|\a)\big)= k\\							\{\O\} &\mbox{if}\ \ z_i\cdot\big(V_\C\cdot P_\C(s_i|\a)- V_\O\cdot P_\O(s_i|\a)\big)< k\ . \end{cases} \label{a(...)}
\end{equation}

Note that (\ref{a(...)}) prescribes the same condition for both $\eps_i\in\BG$, so good and bad hats behave identically conditional on their characteristics $(s_i,z_i)$. However, behavior diverges between ethical types because their $z_i$ is distributed differently, as explained next.
\iffullversion

\fi
First, observe from (\ref{a(...)}) that, as before, every hat of every type targets one of the two systems in \e. Conditional on any given $s_i$, 
\begin{enumerate}[label=(\roman*)]
\item If $V_\C\cdot P_\C(s_i|\a)>V_\O\cdot P_\O(s_i|\a)$, there is a cutoff $z^*(s_i)$ such that  hats with $(s_i,z_i)$ target $\C$ if $z_i>z^*(s_i)$ and target $\O$ if  $z_i<z^*(s_i)$. \label{itm:generalzstar:cutoff}

\item If $V_\C\cdot P_\C(s_i|\a)\le V_\O\cdot  P_\O(s_i|\a)$, all hats with $s_i$ target $\O$ because  $$z_i\cdot(V_\C\cdot P_\C(s_i|\a)- V_\O\cdot P_\O(s_i|\a))<k$$ for all $z_i$. In this case, we set $z^*(s_i)=\infty$. \label{itm:generalzstar:infty}
\end{enumerate}

All hats with skill levels $s_i$ for which case \ref{itm:generalzstar:infty} applies target $\O$ regardless of their other characteristics $(\eps_i,z_i)$. Hence, the ratio of $\B$-hats to $\G$-hats who target $\O$ is $\beta$. All hats $i$ with every other skill level $s_i$, for which case \ref{itm:generalzstar:cutoff} applies, target $\C$ if $z_i>z^*(s_i)$  and target $\O$ if $z_i<z^*(s_i)$ regardless of their \nature $\eps_i$. Hence, since $H_\B$ first-order stochastically dominates $H_\G$, the $\B$ to $\G$ ratio of those who target $\C$ exceeds $\beta$ and the $\B$ to $\G$ ratio of those who target $\O$ is below $\beta$. This means that, conditional on either system having a winner, the likelihood that the winner is a bad hat is strictly higher than $\beta$ in the closed system but is lower than $\beta$ in the open system, establishing part \ref{itm:thm:conditional-on-winner} of \cref{thm:main}.

For parts \ref{itm:thm:lambda} and \ref{itm:thm:lambda-and-tail}, we use the following properties that hold for large $\lambda$.
\begin{observation}\label{obs:biglambda}
It holds that
\begin{equation}
P_\O(s_i|\a)\to 0,\quad P_\C(s_i|\a)\to 0,\quad\mbox{and}\quad z_i^*(s_i)\to\infty\quad\forall s_i\quad\text{as}\quad\lambda\to\infty\ . \label{largelda}
\end{equation}
\end{observation}
All hats for whom $z_i<k/V_\C$ target system $\O$ because their payoff from targeting $\C$ is at most $z_i\cdot V_\C-k<0$. Hence, the expected number of hats who target $\O$ increases without bound as the mean total number of hats increases, i.e., as $\lambda\to\infty$, and so the probability of any given hat winning system $\O$ vanishes, i.e., $P_\O(s_i|\a)\to 0$ for all $s_i$.
Then, unless $z_i^*(s_i)=\infty$,
\begin{equation}\label{zstarPc}
    z_i^*(s_i)\cdot P_\C(s_i|\a)\to k/V_\C > 0\quad\mbox{as}\quad \lambda\to\infty\ ,
\end{equation}
because $z^*_i(s_i)\cdot\big(V_\C\cdot P_\C(s_i|\a)- V_\O\cdot P_\O(s_i|\a)\big)= k$ by definition of $z_i^*$ (see \ref{itm:generalzstar:cutoff}).  As $\lambda\to\infty$, if $z^*_i(s_i)$ stays bounded above for a non-negligible set of $s_i$, the number of hats who have $z_i>z^*(s_i)$ and therefore target $\C$ (see \ref{itm:generalzstar:cutoff}) grows without bound, meaning that  each such hat's probability of winning $\C$ vanishes (i.e., $P_\C(s_i|\a)\to 0$), contradicting \eqref{zstarPc}. Hence, $z_i^*(s_i)\to\infty$ which, given \eqref{zstarPc}, further implies $P_\C(s_i|\a)\to 0$, establishing \cref{obs:biglambda}.\footnote{The argument holds for many standard distributions of $I$, including Poisson, which do not change strangely as the mean ($\lambda$) changes.\label{fn:poisson}}

\iffullversion\medskip\fi

Thus, as $\lambda$ increases, both systems have a winner almost certainly (because the expected number of hats targeting each system grows without bound), more so for the open system since $z_i^*(s_i)\to\infty$. 
By part \ref{itm:thm:conditional-on-winner}, the winner is more likely to be a $\G$-hat in the open system $\O$ (conditional on $\O$ having a winner) than in the closed system $\C$ (conditional on $\C$ having a winner). Since the open system is more likely than the closed system to have a winner, it follows that the unconditional probability for a system to be won by a good hat is higher in the open system than in the closed system, proving part \ref{itm:thm:lambda}. 

On the other hand, although the closed system is more likely to be won by a $\B$-hat conditional there being a winner (that is, \ref{itm:thm:conditional-on-winner}), it is less likely to have a winner than the open system. Since both systems have a winner with probability converging to 1 as $\lambda$ increases, part \ref{itm:thm:lambda-and-tail} follows so long as the gap between the conditional probabilities that the winner is of \nature $\B$ for the open and closed systems does not vanish as $\lambda\to\infty$. This is guaranteed if \eqref{uppertail2} holds.  
\qqed
\end{proof}

\section{Generalized Model II: Many Systems} \label{sec:model-general-manysystems}

Next, we present a further generalized model that is the same as in \cref{sec:model-general}, except that there is a finite number of systems, with each system $\sys$ characterized by an intrinsic value $V_\sys >0$ and a cost $k_\sys\ge0$, and no two systems have identical characteristics $(V_\sys,k_\sys)$.\footnote{Our results hold with identical systems so long as they are treated equally by every hat, in particular, each player targets identical systems with the same probability.} Let $\OSet$ denote the set of open systems, i.e., with $k_\sys=0$, and  $\CSet$ the set of closed systems, i.e., with $k_\sys>0$, both nonempty.
Fix an \e strategy, represented by $\a:\BG\times S\times Z\ \to\ \Delta\big(\OSet\cup\CSet\cup\{\phi\}\big)$ adopted by all players.
Let $\Win_\sys^\eps$ be the probability that an $\eps$-hat wins $\sys$, and 
$\Win_\sys:=\Win_\sys^\B+\Win_\sys^\G$ the probability that $\sys$ has a winner. Define 
\[
 \WWin_{\syss}:= 
 \frac{\sum_{\sys\in\syss}V_\sys\cdot\Win_\sys^\B}
      {\sum_{\sys\in\syss}V_\sys\cdot \Win_\sys} \qquad\mbox{for}
 \qquad \syss\in\OCSet.
\]
\begin{theorem}\label{thm:main2}
    For every $\lambda>0$, there exists an \e in the multi-system generalized game described above. In every equilibrium in which some
closed system is targeted with positive probability,
\[
 \textup{(a)}\quad \WWin_{\CSet}>\beta>\WWin_{\OSet}.
\]
If \eqref{uppertail2} holds, there exists $\lambda_*<\infty$, 
such that, for every \e with $\lambda\geq\lambda_*$,
$\Win_\sys>0$ for all $\sys\in\OSet\cup\CSet$ and, for every $\O\in\OSet$ and $\C\in\CSet$,
\[
 \textup{(b)}\quad \Win_\O^\G>\Win_\C^\G,
 \qquad\quad
 \textup{(c)}\quad \Win_\C^\B>\Win_\O^\B, 
 \qquad\quad
 \textup{(d)}\quad \frac{\Win_\O^\G}{\Win_\O}>\frac{\Win_\C^\G}{\Win_\C}\ .
\]
 \end{theorem}

Parts (a) and (d) of \cref{thm:main2} modify \cref{thm:main}\ref{itm:thm:conditional-on-winner} as follows: (a) means
$\B$-hats are expected to win a strictly larger fraction of total system value won across closed systems, than across open systems; and if
\eqref{uppertail2} holds (d) extends \cref{thm:main}\ref{itm:thm:conditional-on-winner} in a pairwise system-to-system comparison.
Parts (b) and (c) of \cref{thm:main2} respectively extend \cref{thm:main}\ref{itm:thm:lambda}--\ref{itm:thm:lambda-and-tail} to hold for every pair of systems, one open and one closed; however, the condition \eqref{uppertail2} from \cref{thm:main}\ref{itm:thm:lambda-and-tail} is now also needed for (b).
\iffullversion

\fi
The proof of Theorem \ref{thm:main2} is in \cref{app:proof2}. 

\iffullversion
    \subsection{Further generalizations}

    Our models are further generalizable in a number of ways.
    
    \begin{remark}[Multiple winners \& extension over time]\label{rmk:multiple-winners}
The simplest conceptualization of our models is that each system has one vulnerability of expected value $V_\sys$. More realistically, systems may have any number of vulnerabilities, and $V_\sys$ can be thought to be the expected value of reportable or exploitable vulnerability/ies that will be found by the first finder in $\sys$ within the scope of a given game. Our models and results generalize naturally to cover multiple winners (i.e., first finders), who may find the same or different vulnerabilities in the same system at the same time, by adding a multiplicative factor $\eta>1$ to the weighted contest-success probability \eqref{p(s)} as follows:
\begin{equation}
p_i(s_1,\cdots,s_n)\ :=\ \eta\cdot\frac{w(s_i)}{\sum_{j=1}^nw(s_j)}\ \in\ (0,\eta]\ .
\label{p(s)-adjusted}
\end{equation}
An interpretation is that there are multiple vulnerabilities in the system, each with an expected intrinsic value of $V_\sys$, and $p_i(s_1,\cdots,s_n)$ is the expected number of vulnerabilities that hat $i$ would find before others.

Our models treat a system as ``out of play'' once it is won. In reality, of course, players may continue examining it and find more vulnerabilities. Searching for vulnerabilities in systems over an extended period of time can be modeled by a sequence of games like ours, with updated $V_\sys$ values each time a system is won. 
Given the progress made in the targeted system and the lowered intrinsic values of the system just won, the hats may continue to target the same system until all vulnerabilities are found in that system. Those who choose to change targets would choose a new target based on the same principle as in the initial game. Hence, we expect this would not change the qualitative results.
\end{remark}

\begin{remark}[Poisson distribution]
\label{rmk:poisson}
The Poisson assumption on the number of hats is not essential (see \cref{fn:poisson}). It mainly simplifies calculations and helps establish equilibrium existence. The key assumptions are higher costs to examine closed systems and systematically higher rewards from winning systems for bad hats. We intentionally keep good and had hats identical in all other respects, but significant differences in other respects may strengthen or weaken the results.
\end{remark}

\begin{remark}[Targeting multiple systems]\label{rmk:multiple-targets}
Suppose that each hat can target as many systems as they like. Then, all hats of all types target all open systems in \e, so every open system has a winner for sure, and a $\B$-winner with probability $\beta$. Each closed system $\sys$ is contested independently, so there is a decreasing threshold $z^*_\sys(s_i)$ such that the expected payoff from targeting $\sys$ is zero if $z_i= z^*_\sys(s_i)$. Hence, hats target $\sys$ if $z_i> z^*_\sys(s_i)$ but not if $z_i< z^*_\sys(s_i)$. Then, every closed system is won by $\B$ with probability higher than $\beta$ as $\lambda\to\infty$ and our results continue to hold.
\end{remark} 
    \cref{app:thm2-remarks} discusses how our models generalize to certain bounded $z_i$ (\cref{rmk:bound-zi}) and, under practical conditions, to open systems having nonzero entry costs (\cref{rmk:open-sys-entry-costs}).
\fi

\section{Discussion and Open Directions}\label{sec:discussion}

\paragraph{AI}
Recent advances in generative-AI-based cybersecurity tools have scaled up automated and AI-assisted vulnerability detection to unprecedented efficiency, as well as made vulnerability discovery accessible to a wider range of people \cite{cyberscoop-ai-vuln}. 
Human developers remain responsible for triaging results, verifying complex edge cases, and deploying stable fixes --- a manual workflow that is struggling to keep pace with automated discovery  \cite{cyberscoop-ai-vuln,github-ai-cybersec}. As a result, there is a growing backlog of unpatched security issues and an associated increase in exploitation risk \cite{Willison-rumour-bug,Madhavapeddy-rumour-bug}.

One obvious consequence may be that \emph{the number of hats grows}, as vulnerability discovery has become accessible to more people. 
Our explicit modeling treats what happens after a vulnerability is discovered as out of scope. However, the way we model the motivation and payoff of good hats depends intrinsically on hats' assumed motivations (see \cref{sec:model-overview}).
One possibility is that the utility gained by good hats who report their findings would go down due to the decreased likelihood of issues getting fixed.  Another conceivable outcome would be that bad hats' utility upon winning increases due to AI reducing the cost of operationalizing vulnerabilities (although this might be offset by lowered market prices due to increased supply). Either of these adjustments to our model parameters would only strengthen our results in their support of the established folklore.

An explicit modeling of the changes to vulnerability discovery incentives due to recent advances in AI cybersecurity tools, and likely corresponding security outcomes, may be interesting future work.

\paragraph{Other ``security by obscurity'' rationales}
Our work investigates one among several rationales against ``security by obscurity'' --- namely, that open scrutiny facilitates the finding and fixing of problems. Game-theoretic analyses of the other rationales would be interesting future work.

\section{Conclusion}\label{sec:conc}

We have conducted the first rigorous modeling of ``security by obscurity.'' Our game-theoretic analysis provides the first theoretical substantiation of the security community's longstanding experience that hiding how a system works is generally worse for security than inviting widespread scrutiny. Our work contributes to foundations of security economics, as well as practical pushback against industry narratives that draw upon the intuitive appeal of ``security by obscurity.''

\ifanon\else
    \paragraph{\bf Acknowledgments}
    We are grateful to Madiha Zahrah Choksi, Thibaut Horel, and Michael A. Specter for helpful discussion and references.
    
\fi

\clearpage
\printbibliography 

\appendix

\section{Deriving $P_\sys$}\label{app:Psys}

In this section, we derive an expression for $P_\sys$ as defined in \cref{sec:model-general}.
The following observation follows from \ref{itm:generalzstar:cutoff} and \ref{itm:generalzstar:infty} from the proof of \cref{thm:main}.

\begin{observation}\label{obs:strategy}
Any \e strategy $\a$ must have the  property that
any hat with skill $s_i$ will target $\O$ if $z_i<z^*(s_i)$ and will target $\C$ if $z_i>z^*(s_i)$.
\end{observation}

In particular, strategies are deterministic except when $z_i=z^*(s_i)$. Hence, the winning probability of any hat $i$ targeting a system $\sys$, given skill level $s_i=\tilde{s}$,~is:

\[
P_\sys\!(\tilde{s}|\a)\ :=\ \E_{I,\,(s_j)_{j\in[I]}} \left[\sum_{I=0}^\infty\left(\frac{w(\tilde{s})}{w(\tilde{s})+\sum_{j=1}^I \One_\sys(\theta_j)\cdot w(s_j)}\right)\right]
\]
where  the expectation is taken over the distribution of the types of the $I\sim \Poisson(\lambda)$ hats other than $i$\footnote{\label{fn:env eqv}In a game where the total number of players is $I\sim \Poisson(\lambda)$, each player in the game perceives the number of other players as $I\sim \Poisson(\lambda)$, a result called {\it environmental equivalence} due to Myerson \cite{myerson98}.}
and $\One_\sys$ is the indicator function of types that target $\sys$ according to $\a$ (e.g., $\One_\O(\theta_j)=1$ if $z_j\le z^*(s_j)$ and 0 if $z_j>z^*(s_j)$).\footnote{However players behave when $z_i=z^*(s_i)$, as long as they target some system (rather than none) and the conditions in \cref{obs:strategy} hold, the resulting strategy profile is an equilibriunm because (1) when equality holds, targeting either is optimal and (2) a player's behavior when the equality holds does not affect others' optimal behavior, since $z_i=z^*(s_i)$ occurs with measure-zero probability as $H_\B,H_\G$ are atomless. To make the indicator function well defined, we assume without loss of generality that players target $\O$ when the equality holds.}

\section{Alternative Model of Ethical Types}\label{app:alt-model}

As discussed in \cref{sec:model-overview}, while our main models assume players of two types, $\B$ and $\G$, that differ intrinsically, our modeling approach is also able to capture the alternative approach of modeling all players as choosing between ethically reporting or malicious exploiting their findings upon discovery, based on the expected utility from the two options. Here, we present this alternative approach in more detail, and explain how it is captured by our models.

\paragraph{\bf Alternative Model of $(\eps_i,z_i)$}
Each hat $i$, if they win any system $\sys$ of intrinsic value $V_\sys>0$, chooses whether to report their finding for social good or exploit it. The reward they derive from the two options are $r_\G\cdot V_\sys$ and $r_\B\cdot V_\sys$ respectively, where $r_\G$ and $r_\B$ are hat $i$'s private characteristics, i.i.d draws from atomless CDFs $L_\G$ and $L_\B$ with corresponding densities $\ell_\G$ and $\ell_\B$ on $\Re_+$, respectively. Then, upon winning any system, hat $i$ would report their finding if $r_\G>r_\B$ or exploit it if  $r_\G<r_\B$. (The cost $k$  for system $\C$ is sunk at this point, hence irrelevant.) Call hat $i$ a good hat if the former and a bad hat if the latter.  Define a random variable $z:=\max\{r_\G,r_\B\}$ and let $H_\eps$ be the CDF of $z$ conditional on $z=r_\eps$ for $\eps\in\BG$. Suppose
\begin{equation}
\frac{\ell_\B(r)/L_\B(r)}{\ell_\G(r)/L_\G(r)}\ \text{ is nondecreasing and nonconstant on} \ (0,\infty). \label{AltC}
\end{equation}
The lemma below shows that this alternative model is equivalent to the previous modelling from \cref{sec:model-general,sec:model-general-manysystems} of player's type $\theta_i=(\eps_i,s_i,z_i)\in\BG\times S\times Z$, where $s_i\sim F$ independently of $(\eps_i,z_i)$.

\begin{lemma}
In the Alternative Model above, a hat $i$ is bad ($\eps_i=\B$) with probability $\beta:= \Pr[r_\G<r_\B]$ and good ($\eps_i=\G$) with probability $1-\beta$, and derives a reward of $z_i\cdot V_\sys$ from winning system $\sys$ where $z_i\sim H_{\eps_i}$. Moreover, $H_\B$ and  $H_\G$ satisfy \eqref{fosd} and \eqref{uppertail2} from \cref{sec:model-general}.
In other words, the Alternative Model reduces to the models of players' types in \cref{sec:model-general,sec:model-general-manysystems}.
\end{lemma}

\begin{proof}
We must prove \eqref{fosd} and \eqref{uppertail2}.     
We have $\beta = \int_0^\infty L_\G(r)\ell_\B(r)\,dr\in(0,1)$, and
\begin{eqnarray*}
&&H_\B(z)=\frac{1}{\beta}\int_0^z L_\G(r)\ell_\B(r)\,dr\quad\Longrightarrow\quad 
h_\B(z)=\frac{L_\G(z)\ell_\B(z)}{\beta}\ , \\
&&H_\G(z)=\frac{1}{1-\beta}\int_0^z L_\B(r)\ell_\G(r)\,dr\quad\Longrightarrow\quad 
h_\G(z)=\frac{L_\B(z)\ell_\G(z)}{1-\beta}\ .
\end{eqnarray*}
Defining
\begin{equation}
R(z):=\frac{h_\B(z)}{h_\G(z)}
=\frac{1-\beta}{\beta}\cdot\frac{L_\G(z)\ell_\B(z)}{L_\B(z)\ell_\G(z)}\ ,
\label{eq:R}
\end{equation}
 we have $\int_0^z R(r)\,dH_\G(r)=\int_0^z h_\B(r)\,dr=H_\B(z)$ and also for any $z>0$, 
\begin{equation}\begin{aligned}\label{eq:alt-model-convex-combo}    
H_\G(z)\cdot\E\big[R(x)\big|x\sim H_\G,\, x<z\big] +(1-H_\G(z))\cdot\E\big[R(x)\big|x\sim H_\G,\, x>z\big] \\
\ =\ \int_0^\infty R(r)\,dH_\G(r)\ = H_\B(\infty)\,=\,1\,.\hskip77pt
\end{aligned}\end{equation}
Hence, this convex combination of the two expectations in \eqref{eq:alt-model-convex-combo} is 1. 
Since $R(r)$ is nondecreasing and nonconstant by \eqref{AltC}, it follows that
\[
\E\big[R(x)\big|x\sim H_\G,\, x<z\big]\ <\ 1\ <\ \E\big[R(x)\big|x\sim H_\G,\, x>z\big]\quad\ \forall z>0\ .
\]
Consequently, 
\[
H_\B(z)=\int_0^z R(r)\,dH_\G(r)
= H_\G(z)\cdot\E\big[R(x)\big|x\sim H_\G,\, x<z\big]\,<\,H_\G(z)\quad\ \forall z>0\ ,
\]
proving \eqref{fosd}, and $\lim_{z\to\infty}R(z) >\E\big[R(x)\big|x\sim H_\G,\, x>z\big]>1$, proving \eqref{uppertail2}. \qqed
\end{proof}

\section{Proof of \cref{thm:main} and \cref{thm:main2}}\label{app:proof2}

We restate \cref{thm:main2} here for easy reference.
The proof of \cref{thm:main2} below also proves \cref{thm:main}.

\begin{customthm}{\ref{thm:main2}}
    For every $\lambda>0$, there exists an \e in the multi-system generalized game described above. In every equilibrium in which some
closed system is targeted with positive probability,
\[
 \textup{(a)}\quad \WWin_{\CSet}>\beta>\WWin_{\OSet}.
\]
If \eqref{uppertail2} holds, there exists $\lambda_*<\infty$, 
such that, for every \e with $\lambda\geq\lambda_*$,
$\Win_\sys>0$ for all $\sys\in\OSet\cup\CSet$ and, for every $\O\in\OSet$ and $\C\in\CSet$,
\[
 \textup{(b)}\quad \Win_\O^\G>\Win_\C^\G,
 \qquad\quad
 \textup{(c)}\quad \Win_\C^\B>\Win_\O^\B, 
 \qquad\quad
 \textup{(d)}\quad \frac{\Win_\O^\G}{\Win_\O}>\frac{\Win_\C^\G}{\Win_\C}\ .
\]
 \end{customthm}

\begin{proof}
Fix an equilibrium strategy $\a:\Theta\to \Delta\big(\OSet\cup\CSet\cup\{\phi\}\big)$, adopted by all hats. Given the distribution of $\theta$ and the number of hats $I\sim\Poisson(\lambda)$, $\a$ determines, for each system $T$, the distribution of $W_\sys$, the total weight of
other players targeting $T$. Thus, targeting a system $T$, any given hat with skill $s\in S$ wins it with probability 
\begin{equation*}
 P_\sys(s)\ :=\ \mathbb E\!\left[\frac{w(s)}{w(s)+W_\sys}\right]\ \in(0,1]
 \label{eq:v7-payoffs*}
\end{equation*}
and obtains, depending on the hat's type $\theta=(\eps,s,z)$, an expected payoff of 
\begin{equation*}
U_\sys(s,z|\a)\ :=\ zA_\sys(s)-k_\sys\quad\text{where}\quad A_\sys(s):=V_\sys P_\sys(s).
 \label{}
\end{equation*}
Since $U_\sys(s,z|\a)>0$ for any $T\in\OSet$, every hat of every type targets one system that maximizes $U_\sys(s,z|\a)$, called an ``optimal'' system. Since $U_\sys(s,z|\a)$ is independent of $\eps$, so is the set of optimal systems, defined as $\syss^*(s,z)=\arg\max U_\sys(s,z|\a)$. 

We prove (a)--(d) assuming that an \e exists, then prove existence.

\medskip
\underbar{Proving that $\syss^*(s,z)$ is a singleton for a.e.~$(s,z)$ and a cutoff $z^*(s)$ exists.}

We first show that there is a unique optimal system to target for almost everywhere\footnote{``Almost everywhere'' means everywhere except for a measure-zero set.} (henceforth, ``a.e.'') $(s,z)\in S\times Z$. Define the auxiliary functions
\begin{equation*}
 G_\sys(x):=V_\sys\cdot \E\!\left[\frac1{1+xW_\sys}\right],\quad x\geq0,
 \qquad \text{so that}\quad A_\sys(s)=G_\sys\Big(\frac{1}{w(s)}\Big)\ .
 \label{eq:v7-reciprocal}
\end{equation*}
Set $x_0=1/\underline w$ where $\underline w=\min_S w(s)>0$. Expanding by the geometric-series identity and taking expectations term by term (by bounded convergence) gives the power-series representation:
\[
 G_\sys(x)=V_\sys\cdot\sum_{n=0}^{\infty}(-1)^n
 \cdot \E\!\left[\frac{(W_\sys)^n}{(1+x_0W_\sys)^{n+1}}\right](x-x_0)^n,
 \qquad |x-x_0|<x_0\ .
\]
Indeed, $W_\sys/(1+x_0W_\sys)\leq1/x_0$ bounds the absolute value of the
$n$th summand before expectation by
$V_\sys(|x-x_0|/x_0)^n$, permitting termwise expectation.
This single convergent power series represents $G_\sys$ on the open interval
$(0,2x_0)$, which contains every relevant reciprocal weight
$x\in[1/\overline w,1/\underline w]$ where $\overline w=\max_S w(s)$.
Since $W_\sys<\infty$ almost surely, bounded convergence gives
$\lim_{x\downarrow0}G_\sys(x)=V_\sys$. 

Consider any two distinct systems $\sys$ and $\sys'$. If $k_\sys\ne k_{\sys'}$, the payoff difference
$z\cdot(A_\sys(s)-A_{\sys'}(s))-(k_\sys-k_{\sys'})$ is an affine function of $z$ with a
nonzero intercept, so it vanishes at at most one motivation at each skill. 
If $k_\sys=k_{\sys'}$ then $\lim_{u\downarrow0}(G_\sys-G_{\sys'})(u)=V_\sys-V_{\sys'}\ne0$ because $V_\sys\ne V_{\sys'}$, so the analytic function $G_\sys-G_{\sys'}$ is not
identically zero. Its zeros are consequently isolated,\footnote{A real-analytic function on an open interval that is not identically zero has isolated zeros \cite[Corollary 1.2.7 (identity theorem)]{KrantzParks02}.}
and hence finite in the compact interval
$[1/\overline w,1/\underline w]$.
The injectivity of $w$ makes ties of $A_\sys(s)$ (i.e., $A_\sys(s)=A_{\sys'}(s)$)
possible at only finitely many skills across all system pairs, an
$F$-null set. 
Recall that definitionally, $A_\sys(s)$ is the slope of $U_\sys(s,z|\a)$.
At every other skill, any two payoff lines $zA_\sys(s)-k_\sys$
intersect at most one motivation level $z$. Thus the possible motivation levels for which a tie arises
form a finite set at each such skill. The atomlessness of both $H_\eps$ gives
this set probability zero; integrating over $F$ proves that $\syss^*(s,z)$ is a singleton a.e.~$(s,z)$, giving the same unique optimal target for both $\eps\in\BG$ at a given $(s,z)$.

Write $\rho_\sys(s,z)=1$ if the hat targets $\sys$ and $0$ otherwise. The best open
system has slope $A_{\OSet}(s)=\max_{O\in\OSet}A_\O(s)>0$.
A closed system is targeted only if
$zA_\C(s)-k_\C \ge zA_{\OSet}(s)$, hence only if $A_\C(s>A_{\OSet}(s)$ and
$z\ge k_\C/(A_\C(s)-A_{\OSet}(s))$.  Thus, the  cutoff motivation level above which
some closed system is targeted is
\begin{equation*}
 z^*(s):=\min_{\{C\in\CSet:\,A_\C(s)>A_{\OSet}(s)\}}
             \frac{k_\C}{A_\C(s)-A_{\OSet}(s)}\ \ge\ \frac{\min_\C k_\C}{\max_\sys V_\sys}\ >\ 0
 \label{eq:v7-cutoff}
\end{equation*}
if $A_\C(s)>A_{\OSet}(s)$ for some $c\in\CSet$ and $z^*(s):=\infty$ otherwise; 
and hats target an open system if $z<z^*(s)$.
The winning probabilities $P_\sys(s)$, hence the \e conditions, are intact however the hats mix among optimal targets when $\syss^*(s,z)$ is not a singleton, because this happens with probability zero. Hence, we only consider $(s,z)$ for which $\syss^*(s,z)$ is a singleton below, unless stated otherwise.

\medskip
\underbar{Proving part (a)} 

The Poisson {\it add-one formula}, a special case of the Mecke
equation,\footnote{See Last and Penrose \cite[Theorem 4.1,
equation (4.2), p.~27]{LastPenrose17}.}
equates the expected sum of players' contributions to $\lambda$ times
the expected contribution of an independently sampled additional player
facing the original Poisson population. Applying this to a player's
contribution to the probability of winning a system $\sys$, for each $\eps\in\BG$,
\begin{equation}
 \Win_\sys ^\eps=\lambda \beta_\eps\int_S\int_0^\infty
       \rho_\sys (s,z)P_\sys(s)\,dH_\eps(z)\,dF(s).
 \label{eq:v7-winner-identity}
\end{equation}
Here $\beta_\B=\beta$ and $\beta_\G=1-\beta$. For open systems $O\in\OSet$, 
\begin{equation}
\Win_\O^\eps=\lambda \beta_\eps\int_S P_\O(s)H_\eps(z^*(s))\,dF(s)
 \quad\Longrightarrow\quad  \Win_\O^\B\leq \beta\cdot \Win_\O
 \label{eq:v7-open-bound}
\end{equation}
because, given $H_\B< H_\G$, the integral for $\eps=\B$ is smaller than the integral for $\G$, so
$\Win_\O^\B/\beta\leq\Win_\O^\G/(1-\beta)$, which rearranges to the inequality in \eqref{eq:v7-open-bound}.

Define $A^*(s,z):=\sum_{\sys\in\OCSet}\rho_\sys (s,z)A_\sys (s)$ which captures the slope of $U_\sys (s,z|\a)$ for the targeted system, which is the slope of the upper envelope of $U_\sys (s,z|\a)$ across all systems. Restricting to $z<z^*(s)$ where an open system is targeted, we define $A^*_{\OSet}(s,z):=A_{\OSet}(s)\One_{\{z<z^*(s)\}}$, while restricting $z>z^*(s)$ we define $A^*_{\CSet}(s,z):=\sum_{\sys\in \CSet}\rho_\sys (s,z)A_\sys (s)$ which is a nondecreasing step function, zero below $z^*(s)$
and positive above.
Let $\mathbb E_\eps$ average over $z\sim H_\eps$ at fixed $s$.
Strict stochastic dominance gives
$\mathbb E_\B[A^*_{\CSet}(s,z)]\geq\mathbb E_\G[A^*_{\CSet}(s,z)]$ and
$\mathbb E_\B[A^*_{\OSet}(s,z)]\leq\mathbb E_\G[A^*_{\OSet}(s,z)]$, 
strictly for both when $z^*(s)<\infty$.
Equation \eqref{eq:v7-winner-identity} gives
\[
 \sum_{\sys\in\syss}V_\sys \Win_\sys ^\eps
 =\lambda \beta_\eps\int_S\mathbb E_\eps\big[A^*_{\syss}(s,z)\big]\,dF(s),\qquad {\syss}\in\OCSet
\]
so normalizing by 
$$\sum_{\sys\in\syss}V_\sys \Win_\sys=\sum_{\sys\in\syss}V_\sys \Win_\sys ^\B+\sum_{\sys\in\syss}V_\sys \Win_\sys ^\G\ ,$$
\[
\WWin_{\CSet}
= \frac{\beta\int_S\mathbb E_\B\big[A^*_{\CSet}(s,z)\big]\,dF(s)}{\beta\int_S\mathbb E_\B\big[A^*_{\CSet}(s,z)\big]\,dF(s)+(1-\beta)\int_S\mathbb E_\G\big[A^*_{\CSet}(s,z)\big]\,dF(s)} > \beta
\]
when $\syss=\C$ and the reverse inequality for $\WWin_{\OSet}$ when $\syss=\O$, proving part~\textup{(a)}. This also proves part \ref{itm:thm:conditional-on-winner} of \cref{thm:main} because it is equivalent to part (a) of \cref{thm:main2} when there are one open and one closed system.  

\medskip
\underbar{Proving parts (b)--(d)}\label{pf:main2:b-d}

Let $m_\sys $ be the mean number of hats targeting $\sys$. Poisson thinning~\cite{Kingman92} gives
$N_\sys \sim\mathrm{Poisson}(m_\sys )$ other hats targeting $\sys$ and $\Win_\sys =1-e^{-m_\sys }$. Moreover,
\begin{equation}
 e^{-m_\sys }\leq P_\sys (s)\leq
 \gamma\cdot\mathbb E\!\left[\frac1{1+N_\sys }\right]
 =\gamma\cdot\frac{1-e^{-m_\sys }}{m_\sys }\quad\text{where}\quad \gamma:=\frac{\overline{w}}{\underline{w}},
 \label{eq:v7-congestion-bound}
\end{equation}
with $(1-e^{-m_\sys })/m_\sys =1$ at $m_\sys =0$. The lower bound $e^{-m_\sys }$ counts the event of a player targeting $\sys$ alone, on which the player wins surely; conditional on $N_\sys =n$,
the player's weight share is at most $\gamma/(1+n)$, proving the upper bound.

We show that every $m_\sys$ diverges uniformly over equilibria: that is, for any $L\in\Re_+$, there is a threshold $\lambda_L$ such that if $\lambda>\lambda_L$, then $m_\sys>L$ for every $\sys$ in every equilibrium.
If, to the contrary, $m_{\sys_0}\leq L$ for some system $\sys_0$ and $L>0$ for all $\lambda>0$,
then $A_{\sys_0}(s)\geq \nu:=\min_\sys  V_\sys  e^{-L}$ at every skill. 
Set $z_0=1+2(\max_\sys  k_\sys )/\nu$ and $M=2\gamma (\max_\sys V_\sys )/\nu$. Every player with $z\geq z_0$ can obtain payoff 
\[
U_{\sys_0}(s,z|\a)=zA_{\sys_0}(s) -k_{\sys_0} \ge z \nu-\nu\frac{z_0-1}{2}\ >\ \frac{z \nu}{2}
\]
by targeting $\sys_0$, so their targeted system must have $A_\sys (s)\geq \nu/2$.
By \eqref{eq:v7-congestion-bound}, that system has $m_\sys \leq M$.\footnote{To be precise, $m_\sys\le M+W(-Me^{-M})< M$ where $W(x)$ is the Lamber $W$ function, the inverse function of $f(w)=we^w$, hence the the last inequality follows from $-Me^{-M}<0$.}
Consequently, 
\[
 \lambda\Pr[z\geq z_0]\
 \leq\sum_{\sys:\,m_\sys \leq M}m_\sys \leq |\OSet\cup\CSet|\cdot M
\]
where $\lambda\Pr[z\geq z_0]$ is the expected number of hats with $z\ge z_0$. Since $\Pr[z\geq z_0]>0$ and both $\Pr[z\geq z_0]$ and $M$ depend only on primitives and $L$, so
$m_{\sys_0}\leq L$ is impossible if $\lambda>|\OSet\cup\CSet|\cdot M/\Pr[z\geq z_0]$ in
every equilibrium. This proves the asserted uniform divergence.
Equation \eqref{eq:v7-congestion-bound} and the cutoff formula then give,
uniformly over equilibria, as $\lambda\to\infty$,
\begin{equation}
 \min_\sys \Win_\sys \rightarrow1,\quad
 \varepsilon_\lambda:=\max_\sys \sup_{s\in S}A_\sys (s)\rightarrow0,\quad
 \inf_{s\in S}z^*(s)\geq\frac{\min_{\C\in\CSet} k_\C}{\varepsilon_\lambda}
 \rightarrow\infty.
 \label{eq:v7-uniform-limits}
\end{equation}

For \cref{thm:main}, i.e., when there is one open and one closed system, $\Win_\O>\Win_\C\to1$ as $\lambda\to\infty$ because $z^*(s)\to\infty$. For the closed system, \eqref{fosd} and \eqref{uppertail2} imply $\Win^\G_\C\leq(1-\beta)\cdot\Win_\C$. Together with \eqref{eq:v7-open-bound}, which is equivalent to $\Win^\G_\O\geq(1-\beta)\cdot \Win_\O$, we deduce that $\Win^\G_\O>\Win^\G_\C$ for $\lambda$ large enough.

\medskip
For \cref{thm:main2}, when there are multiple systems, by \eqref{eq:v7-uniform-limits} eventually every system is targeted with positive probability and every
player targeting a closed system has $z\geq\widehat z$. 
Applying \eqref{uppertail2} to \eqref{eq:v7-winner-identity} for closed systems gives
\[
 \Win_\C^\G\leq\frac{\chi(1-\beta)}{\beta}\Win_\C^\B,
 \qquad \frac{\Win_\C^\B}{\Win_\C}\geq
 \xi:=\frac{\beta}{\beta+\chi(1-\beta)}>\beta.
\]
Together with \eqref{eq:v7-open-bound}, this proves
$\Win_\O^\G/\Win_\O\geq1-\beta>1-\xi\geq\Win_\C^\G/\Win_\C$, hence \textup{(d)}.
Also $\Win_\O^\G\geq(1-\beta)\cdot\Win_\O$, $\Win_\C^\G\leq1-\xi$,
$\Win_\C^\B\geq\xi\cdot\Win_\C$, and $\Win_\O^\B\leq \beta$.
Both $(1-\xi)/(1-\beta)$ and $\beta/\xi$ are below one, so
\eqref{eq:v7-uniform-limits} eventually gives
$\min_\sys  \Win_\sys >\max\{(1-\xi)/(1-\beta),\,\beta/\xi\}$.
This proves \textup{(b)} and \textup{(c)}.

\medskip
\underbar{Proving existence of equilibrium}

Define $H:=\beta\cdot H_\B+(1-\beta)\cdot H_\G$ and $J=|\OSet\cup\CSet|$. Recall $s\sim F$ on $S=[0,1]$.

\paragraph{1. Cumulative targeting probabilities}
For an arbitrary strategy $\a$, let $\mu_\sys(x)$ denote the probability that a hat has skill at
most $x$ and targets $\sys$. 
Equip $C(S)^J$ with $\|\mu\|:=\max_\sys\sup_{s\in S}|\mu_\sys(s)|$ where $C(S)$ is the vector space of all continuous real-valued functions on $S$, and define
\begin{equation}
 \mathcal K=\left\{\mu\in C(S)^J:
 \mu_\sys\ge0,\quad\mu_\sys\text{ is nondecreasing},\quad
 \sum_\sys\mu_\sys=F\right\}\,.
 \label{eq:kfg-domain}
\end{equation}

Since $F(0)=0$, every $\mu_\sys(0)=0$. For $x\le y$,
\begin{equation}
 0\le\mu_\sys(y)-\mu_\sys(x) \le F(y)-F(x)\ . \label{eq:kfg-modulus}
\end{equation}
Thus every $\mu_\sys$ takes values in $[0,1]$ and has modulus of continuity bounded above by that of $F$, and hence form an equicontinuous family.
So, $\mathcal K$ is compact by
Arzel\`a--Ascoli \cite[Theorem~A5]{Rudin1991}.\footnote{Arzel\`a--Ascoli gives relative compactness of $\mathcal K$ in $C(S)^J$, i.e., every sequence $\mu^n$ in $\mathcal K$ has a uniformly convergent subsequence with the limit in $C(S)^J$. The limit is in $\mathcal K$ because the convergence is uniform under the given norm, giving compactness of~$\mathcal K$.}
$\mathcal K$ is convex and nonempty
($\mu_\sys=F/J$ is feasible). 

\paragraph{2. Payoffs and best responses}
For $\mu\in\mathcal K$, let $W_\sys^\mu$ be the sum of $w(s_j)$ of hats $j$ who target system $\sys$ according to $\mu_\sys$. It is finite
 almost surely. As in the earlier payoff definition, set
\begin{equation}
 P_\sys^\mu(s):=\mathbb E\!\left[\frac{w(s)}{w(s)+W_\sys^\mu}\right],\quad
 A_\sys^\mu(s):=V_\sys P_\sys^\mu(s),\quad
 U_\sys^\mu(s,z):=zA_\sys^\mu(s)-k_\sys\ .
 \label{eq:kfg-payoffs}
\end{equation}
Since $\a$ generates $\mu$, we have $U_\sys^\mu(s,z)= U_\sys(s,z|\a)$. 
Hence, by the same argument as before,  $\syss^*(s,z|\mu)=\arg\max U_\sys^\mu(s,z)$
is a singleton for a.e.~$(s,z)$ and a cutoff $z^*(s)$ exists for each $s$ such that
the (unique) optimal target is an open system if $z<z^*(s)$ and a closed system if $z>z^*(s)$.

To prove continuity of $P_\sys^\mu$ in $\mu$, set $g_t(v)=1-e^{-t w(v)}$. The Poisson
Laplace formula and Tonelli's theorem give
\cite[Theorem~3.9]{LastPenrose17} 
\begin{equation}
 P_\sys^\mu(s)=\int_0^\infty w(s)e^{-t w(s)}
 \exp\!\left\{-\lambda\int_Sg_t(v)\,d\mu_\sys(v)\right\}dt\ .
 \label{eq:kfg-laplace}
\end{equation}
Stieltjes integration by parts bounds the change in the inner integral
by $2\|\mu_\sys-\nu_\sys\|_\infty$, since $g_t$ increases within $[0,1]$.
Using $|e^{-a}-e^{-b}|\le|a-b|$ for $a,b\ge0$ in
\eqref{eq:kfg-laplace} yields
\begin{equation}
 \|P_\sys^\mu-P_\sys^\nu\|_\infty
 \le2\lambda\|\mu_\sys-\nu_\sys\|_\infty\ ,
 \label{eq:kfg-payoff-bound}
\end{equation}
proving continuity of $P_\sys^\mu$, hence that of $A_\sys^\mu(s)$ and $U_\sys^\mu(s,z)$.

\paragraph{3. The induced map and a fixed point}
Define the cumulative probabilities that system $\sys$ is targeted, induced by $\mu$:
\begin{equation}
 \Psi_\sys(\mu)(x)=\int_{[0,x]}\int_0^\infty
 \mathbf1_{\{\sys\in \syss^*(s,z|\mu)\}}\,dH(z)\,dF(s)\ .
 \label{eq:kfg-map}
\end{equation}
These are probabilities, so no factor $\lambda$ appears.
The coordinates are nonnegative and nondecreasing, sum to $F$, and
satisfy \eqref{eq:kfg-modulus}. Thus $\Psi(\mu)\in\mathcal K$.
If $\mu^n\to\mu$ uniformly, \eqref{eq:kfg-payoff-bound} gives payoff
convergence at every fixed type. For a.e. $(s,z)$ for which $\syss^*(s,z|\mu)$ is a singleton, $\syss^*(s,z|\mu^n)=\syss^*(s,z|\mu)$ for large enough $n$ due to continuity of expected utility from targeting each system. Hence
\begin{equation}
 \|\Psi(\mu^n)-\Psi(\mu)\|
 \le\int_S\int_0^\infty
 \mathbf1_{\{\syss^*(s,z|\mu^n)\ne \syss^*(s,z|\mu)\}}\,dH(z)\,dF(s)
 \rightarrow0.
 \label{eq:kfg-map-continuity}
\end{equation}
This proves continuity of $\Psi$ in the uniform norm.

Next, we apply the Kakutani--Fan--Glicksberg theorem
\cite[Corollary~17.55]{AliprantisBorder2006} to
$\Gamma:\mathcal K\to2^{\mathcal K}$ defined by
$\Gamma(\mu):=\{\Psi(\mu)\}$, where $2^{\mathcal K}$ denotes the power set of $\cal K$. $\Gamma$ satisfies the conditions of Kakutani--Fan--Glicksberg theorem because $\Psi$ is a continuous function,\footnote{Since $C(S)^J$ is a normed vector space, it is
Hausdorff and locally convex.} so it follows that $\mu^*\in\Gamma(\mu^*)$ for some $\mu^*\in\mathcal K$, or
equivalently, $\mu^*=\Psi(\mu^*)$.

\paragraph{4. From the fixed point to an \e}
For each $\eps\in\BG$, consider 
\begin{equation}
 \a^*(\eps,s,z)\ \in\ \Delta\big(\syss^*(s,z|\mu^*)\big)\ . \label{eq:kfg-strategy}
\end{equation}
The cumulative targeting probabilities generated by $\a^*$ are
\begin{equation}
 \Pr\{s_i\le x,\ \sigma_{\eps_i}^*(s_i,z_i)=\sys\}
 =\Psi_\sys(\mu^*)(x)=\mu_\sys^*(x)\ .
 \label{eq:kfg-consistency}
\end{equation}
Hence, it constitutes an equilibrium for all hats to adopt $\a^*$.
\qqed
\end{proof}

\section{\iffullversion{Additional }\fi{Remarks} on Generalizability}
\label{app:thm2-remarks}

\begin{remark}[Upper bound of $z_i$]\label{rmk:bound-zi}
\cref{thm:main2} extends to atomless motivation distributions with common full support $[0,\bar z]$, if \eqref{fosd} and \eqref{uppertail2} hold on $[0,\bar z]$ and 
\begin{equation}
\bar z\ >\ \max\Bigg\{\widehat z,\,\frac{1}{(1-\beta)(1-\chi)}\cdot\max_{\C\in\CSet}\frac{k_\C}{V_\C} \Bigg\}\ . \label{app:boundz}
\end{equation}

In the proof of \cref{thm:main2},
the arguments for unique choices, cutoffs, the add-one formula, and the weighted average comparison between open systems and closed systems are unaffected, so part (a) remains valid. The uniform congestion argument which led to \eqref{eq:v7-uniform-limits} requires modification. Closed-system loads $m_\C$ (as defined in the \hyperref[pf:main2:b-d]{proof of (b)--(d)} of \cref{thm:main2}) remain bounded:
$m_\C\le\gamma\bar zV_\C/k_\C$ since a larger load would make targeting $\C$ unprofitable even at the highest motivation. 
In contrast, every open-system load $m_\O$ diverges as $\lambda\to\infty$ by the same reasoning as before (i.e., all hats with $z_i < \min_{\C\in\CSet}k_\C/V_\C$ target open systems).
Thus, each open system gets targeted with probability tending to one and open-system payoffs tend to zero.

The bounded total load at closed systems then forces $\sup_s(\bar zA_\C(s)-k_\C)\to0$ for every closed system. Otherwise, uniform continuity of $w(s)$ and full support of the type distributions would give a fixed positive fraction of players a strictly positive closed-system payoff, contradicting bounded closed-system loads as open-system payoffs vanish. Consequently, closed systems are targeted by hats with motivations approaching $\bar z$, so the density condition \eqref{uppertail2} applies to all closed systems eventually.  The individual comparisons argument (in the \hyperref[pf:main2:b-d]{proof of (b)--(d)} of \cref{thm:main2}) therefore still gives a conditional $\B$-share of at least $\xi>\beta$ for closed systems, proving (d) given \eqref{eq:v7-open-bound}. 
The lower bound of \eqref{eq:v7-congestion-bound} implies $A_\C(s)\ge V_\C e^{-m_\C}$, hence $\Win_\C\ge1-k_\C/(\bar zV_\C)$. \eqref{app:boundz} ensures $\xi(1-k_\C/(\bar zV_\C))>\beta$, proving (b) given $\Win_\O^\B\leq \beta$. Then, recalling from the previous paragraph that each open system is targeted with probability tending to one, (c) follows.
\end{remark}

\begin{remark}[Open systems' entry costs]\label{rmk:open-sys-entry-costs} 
The core arguments underlying our results are that (i) more financially motivated hats (who tend to be $\B$) are more likely to target closed systems despite high cost, so the conditional probability that the winner is $\B$ given there is a winner is strictly higher for closed systems, and (ii) the relevant condition (i.e., systems have a winner) is met with a rapidly increasing probability, converging to 1 as $\lambda\to\infty$. Although formally presented as a limit outcome, our results tend to hold for reasonable values of $\lambda$ because the effect (ii) is already strong enough, as shown in \cref{app:simulation}.

A small entry cost $k_o$ for open systems changes the argument (i), especially at the limit as $\lambda\to 1$: because the expected number of hats targeting open systems is now bounded, only hats with increasingly large $z_i$'s enter any system, sending the $\B:\G$ ratio of entrants in every system to the same limit, $\lim_{z\to\infty} h_\B(z)/h_\G(z)$. This may (but does not necessarily)  invalidate our limit result. However, for reasonable ranges of $\lambda$, the effects of (i) and (ii) above tend to be strong enough for the open system to be more secure under sensible parameters. For example, for a large range of $\lambda$ above the bound of $\lambda_*$ approximated in \cref{app:simulation}, the composition of $\B$ and $\G$ hats that target each system changes only slightly when the cost of examining open systems increases slightly above zero, hence our results should remain valid. 
\end{remark}

\section{Approximating the Number of Hats Needed}\label{app:simulation}

Here, we show that the minimum number $\lambda_*$ of hats needed for our results to hold is well within realistic ranges for a range of reasonable parameter settings.
\cref{app:simulation:params} sets parameters, \cref{app:simulation:results} presents our approximation results, \cref{app:simulation:bound} derives the bound that we approximate, and \cref{app:simulation:verif} explains how we computed and verified the approximations.

\subsection{Parameter setting}\label{app:simulation:params}

Consider one open system $\O$, one closed system $\C$, and $I$ hats where $I\sim\Poisson(\lambda)$, with independent player types and fix parameters as follows:
\[
 V_\O=V_\C=1,\qquad \beta=0.2,\qquad k_\O=0,\qquad k_\C=k\in\{0.5,1,2\},\qquad S=[0,1].
\]
$V_\O=V_\C=1$ normalizes system values. $\beta=0.2$ is a conservative estimate for the fraction of bad hats among all hats; assuming a low $\beta$ strengthens our results. $k_\O=0$ in our models by assumption.
We fix the cost $k\in\{0.5,1,2\}$ of examining closed systems arbitrarily such that, given $z$ as in \eqref{eq:illustration-laws} below, not too high a fraction of hats are precluded from targeting $\C$ because their reward falls short of their cost even if they win for sure (i.e., $z\cdot V_\C < k$).
Given the long history and practice of vulnerability research on proprietary and otherwise ``closed'' systems, we believe too high $k$ values are unrealistic as they would mean that the cost of examining closed systems is so high that very few people would examine them.

Recall from \eqref{eq:v7-congestion-bound} in the proof of \cref{thm:main2} that we define $\gamma:=w(1)/w(0)$ as the ratio of the maximum to the minimum possible weight of a given weight function.
Our theoretical derivation of the bound for $\lambda_*$ (\cref{app:simulation:bound}) works for every atomless distribution $F$ of $s_i\in[0,1]$ and every strictly increasing weight function $w:S\to(0,\infty)$ with a given $\gamma=w(1)/w(0)>1$. Our approximations consider $\gamma\in\{1.2,2,5\}$. Because of the generality of the bound, our approximations are very conservative.

For $z\ge0$ and $0\le\tau<1$, let
\begin{equation}
 H_\B(z)=1-e^{-z/2}\ ,\qquad
 H_\G^{(\tau)}(z)=1-\tau e^{-z/2}-(1-\tau)e^{-z}\ .
 \label{eq:illustration-laws}
\end{equation}
That is, we fix a CDF of motivations for bad hats as an exponential distribution and consider a series of motivation CDFs for good hats which are first-order stochastically dominated as in \eqref{fosd}, parametrized by $\tau$. 
We observe that the mean motivations are $2$ for $\B$ and $1+\tau$ for $\G$; as rewards are calculated as $z\cdot V_\sys$, we see that $z$ largely remains within a small constant range (see \cref{fig:smallconstantrange}) that keeps the intrinsic values $V_\O,V_\C$ meaningful as a factor in rewards. We believe this is realistic in that very high rewards are generally likely to be rarer. Note that $H_\G^{(\tau)}$ converges to $H_\B$ as $\tau\to1$. The larger the $\tau$, $H_\B$ and $H_\G^{(\tau)}$ become increasingly similar; since our result depends on the distributions of good and bad hats being different, we expect our result to be stronger (i.e., our $\lambda_*$ bound to be lower) for lower $\tau$, as shown in \cref{tab:illustration-results}.
All pairs $H_\B,H_\G^{(\tau)}$ for $\tau\in[0,1)$ satisfy \eqref{fosd} and \eqref{uppertail2}.

\begin{figure}[t!]
\begin{subfigure}[b]{0.48\textwidth}
        \centering
        \includegraphics[width=\textwidth]{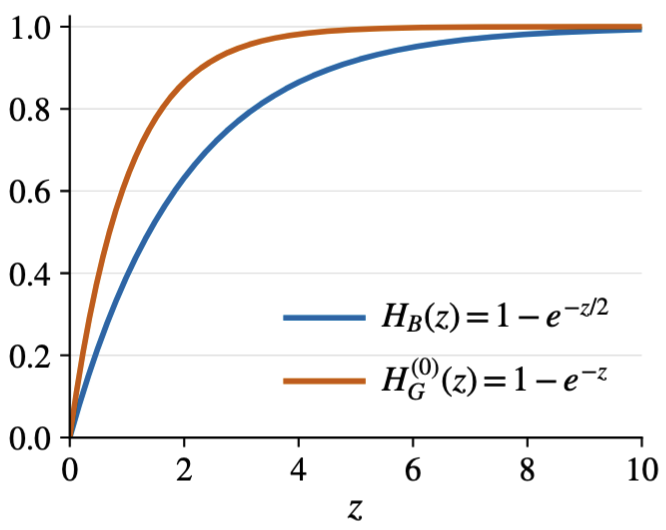} \caption{Distribution functions}
        \label{subfig:dist-fns}
    \end{subfigure}
    \hfill \begin{subfigure}[b]{0.48\textwidth}
        \centering
        \includegraphics[width=\textwidth]{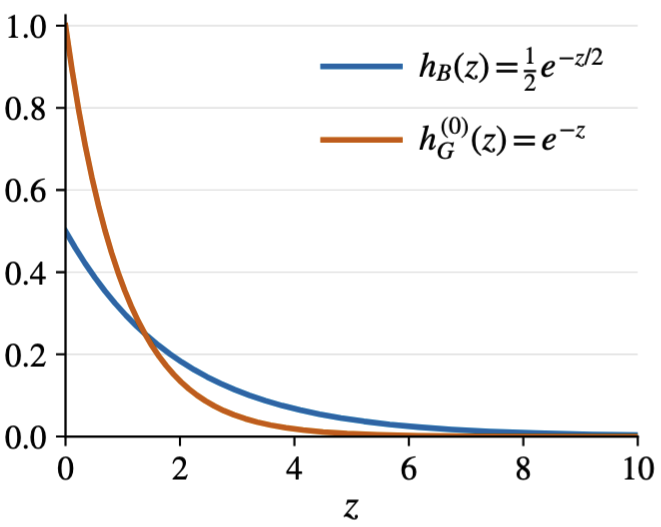} \caption{Densities}
        \label{subfig:densities}
    \end{subfigure}
    \caption{Motivation distributions for our approximations}
\label{fig:smallconstantrange}
\end{figure}

We are interested in $\lambda_*$ such that 
\begin{equation}
\pi_\C^\B>\pi_\O^\B\quad\text{and}\quad \pi_\O^\G>\pi_\C^\G\qquad \forall \lambda \ge\lambda_*\ .
 \label{eq:illustration-comparisons}
\end{equation}

\subsection{Approximation results}\label{app:simulation:results}

\iffullversion\begin{table}[htbp]
\centering
\caption{Conservative upper bound of $\lambda_*$ based on \eqref{eq:illustration-comparisons}, for varying $k$, $\gamma$, and $\tau$}
\label{tab:illustration-results}
\vspace{0.5em}
\renewcommand{\arraystretch}{1.25} \begin{tabular}{c rrr rrr rrr}
\toprule
& \multicolumn{3}{c}{$k = 0.5$} & \multicolumn{3}{c}{$k = 1$} & \multicolumn{3}{c}{$k = 2$} \\
\cmidrule(lr){2-4} \cmidrule(lr){5-7} \cmidrule(lr){8-10}
$\tau$ & {$\gamma=1.2$} & {$\gamma=2$} & {$\gamma=5$} &{$\gamma=1.2$} & {$\gamma=2$} & {$\gamma=5$} & {$\gamma=1.2$} & {$\gamma=2$} & {$\gamma=5$} \\
\midrule
0.00 &  6 & 10 & 28 &  7 &  10  &  24 &   9 &  12 &     25 \\
0.10 &  7 & 10 & 30 & 7 &  11 &    26 &   9 &  13 &     28 \\
0.25 &  7 & 11 & 34 &  8 &  12 &    31 &  11 &  16 &     35 \\
0.50 &  8 & 14 & 46 & 9 &  16 &    46 &  15 &  24 &     63 \\
0.75 & 11 & 19 & 82 & 14 &  26 &   107 &  28 &  55 &    250 \\
0.90 & 15 & 31 & 187 & 24 &  54 &   434 &  74 & 222 &  4,303 \\
\bottomrule
\end{tabular}
\end{table} \fi

Table~\ref{tab:illustration-results} reports upper bounds for $\lambda_*$, for the parameter settings described in \cref{app:simulation:params},
computed by approximation as explained in \cref{app:simulation:bound,app:simulation:verif}.

\begin{remark}
As detailed in \cref{app:simulation:bound}, the upper bounds we use for our approximations are derived from relatively loose conditions (see, e.g., \eqref{eq:illustration-probability-bounds} and \eqref{eq:illustration-uniform-bounds} and definitions of $\Delta_\B(\lambda)$ and $\Delta_\G(\lambda)$), that works for all distributions of skills and weight functions of a given $\gamma$. As a consequence, the reported bounds are very conservative, and actual values of $\lambda_*$ for many parameter settings likely are substantially lower than the bounds in \cref{tab:illustration-results}.
\end{remark}

\subsection{Bounding $\lambda_*$}\label{app:simulation:bound}

Let $q(s):= q_\C(s)$ be the probability of targeting $\C$ conditional on skill in equilibrium, and $q_\O(s)=1-q(s)$. 
Recall
\begin{equation}
 P_\sys(s)=\mathbb E\!\left[\frac{w(s)}{w(s)+W_\sys}\right],\qquad
 m_\sys=\lambda\int_0^1q_\sys(v)\,dv,\qquad m_\C+m_\O=\lambda.
 \label{eq:illustration-payoffs}
\end{equation}
The closed-minus-open payoff is $zD(s)-k$, where $D=P_\C-P_\O$.
Consequently, players target $\C$ above $z^*(s)=k/D(s)$ when $D(s)>0$,
and target $\O$ otherwise; set $z^*(s)=\infty$ if $D(s)\le0$. Hence, $q_\C(s)=Q_\tau(D(s))$, where
\begin{equation}
 Q_\tau(d):=
 \begin{cases}
 \big(\beta+(1-\beta)\tau\big)e^{-k/(2d)} +(1-\beta)(1-\tau)e^{-k/d}& \mbox{~~if~~} d>0\\
 0&\mbox{~~if~~}d\le0\ .
 \end{cases}
 \label{eq:illustration-response}
\end{equation}
From \eqref{eq:v7-winner-identity} and \eqref{eq:kfg-laplace}, since $\int_0^s q_\C(v)dv=\mu_\C(s)$ and $\int_0^s q_\O(v)dv=\mu_\O(s)$, we have
\begin{align}
&\pi_\C^\eps=\lambda\beta_\eps\int_0^1 P_\C(s)\big(1-H_\eps(z^*(s))\big)\,ds,\quad
 \pi_\O^\eps=\lambda\beta_\eps\int_0^1 P_\O(s)H_\eps(z^*(s))\,ds,
 \label{eq:illustration-winners}\\
&P_\sys(s)\,=\int_0^\infty w(s)e^{-t w(s)}
 \exp\!\left\{\lambda\int_0^1q_\sys(v)(e^{-t w(v)}-1)\,dv\right\}dt,
 \label{eq:illustration-laplace}
\end{align}
where $\beta_\B=\beta$ and $\beta_\G=1-\beta$. In particular,
$\pi_\sys^\B+\pi_\sys^\G=1-e^{-m_\sys}$.

Recall $w(0)=1$ and let $\gamma=w(1)>1$, so $1\le w(s)\le\gamma$.
For $N_m\sim\operatorname{Poisson}(m)$, define
\begin{equation}
 P_*(m):=\mathbb E\!\left[\frac1{1+\gamma N_m}\right],\qquad
 P^*(m):=\mathbb E\!\left[\frac{\gamma}{\gamma+N_m}\right].
 \label{eq:illustration-probability-bounds}
\end{equation}
Dependence on $\gamma$ is suppressed. Both functions are continuous,
positive, and nonincreasing, with value one at $m=0$. Conditional on
$N$ opponents, $N\le W_\sys\le\gamma N$, which gives
\begin{equation}
 \frac1{1+\gamma N}\le\frac{w(s)}{w(s)+W_\sys}
 \le\frac{\gamma}{\gamma+N},\qquad
 P_*(m_\sys)\le P_\sys(s)\le P^*(m_\sys)\ .
 \label{eq:illustration-uniform-bounds}
\end{equation}

For fixed $\tau$, define for $v\in [0,\lambda]$
\begin{equation}
 \begin{aligned}
 R_\C(v\,|\,\lambda)&=\lambda\cdot Q_\tau\bigl(P_*(v)-P^*(\lambda-v)\bigr)\  ,\\
 R_\O(v\,|\,\lambda)&=\lambda\cdot\left(1-Q_\tau\bigl(P^*(\lambda-v)-P_*(v)\bigr)\right)\ .
 \end{aligned}
 \label{eq:illustration-load-maps}
\end{equation}
Each $R_\sys$ is continuous and nonincreasing, with values in $[0,\lambda]$.
Thus $R_\sys(v\,|\,\lambda)-v$ is strictly decreasing and changes weakly from
nonnegative to nonpositive at the endpoints; it has a unique zero.
Let $x(\lambda)$ and $y(\lambda)$ be, respectively, the $v$ value that solves $R_\C(v\,|\,\lambda)=v$ and $R_\O(v\,|\,\lambda)=v$.

Since $D(s)\ge P_*(m_\C)-P^*(m_\O)$ from \eqref{eq:illustration-uniform-bounds} and $Q_\tau$ is nondecreasing,  it follows that 
\[
 q_\C(s)=Q_\tau(D(s))
 \ge Q_\tau\!\left(P_*(m_\C)-P^*(m_\O)\right).
\]
The right-hand side is independent of $s$. Integrating over $s$ and using $m_\O=\lambda-m_\C$, we obtain $m_\C\ge R_\C(m_\C;\lambda)$ because
\begin{align*}
m_\C & =\lambda\!\int_0^1\! q_\C(s)\,ds \ge \lambda Q_\tau\!\left(P_*(m_\C)\!-\!P^*(m_\O)\right) \\ & =\lambda Q_\tau\!\left(P_*(m_\C)\!-\!P^*(\lambda-m_\O)\right) =R_\C(m_\C;\lambda)\ ,
\end{align*}
and similarly, $m_\O\ge R_\O(m_\O;\lambda)$.
Hence, in every equilibrium,
\begin{equation}
 m_\C\ge x(\lambda)\ ,\qquad m_\O\ge y(\lambda)\ .
 \label{eq:illustration-load-lower}
\end{equation}
Both $x(\lambda)$ and $y(\lambda)$ are nondecreasing in $\lambda$: for fixed $v$, increasing
$\lambda$ decreases $P^*(\lambda-v)$ and thereby increases both
responses in \eqref{eq:illustration-load-maps}, as well as the prefactor, $\lambda$.

Since $D(s)\le P_\C(s)\le P^*(x(\lambda))$, every cutoff satisfies
$z^*(s)\ge\underline z:=k/P^*(x(\lambda))$. Define
\begin{equation}
 \xi(\lambda)=
 \frac{\beta}{\beta+(1-\beta)
       \big(\tau+(1-\tau)e^{-\underline z/2}\big)} >\beta\ .
 \label{eq:illustration-xi}
\end{equation}
The identity
$1-H_\G(z)=(1-H_\B(z))\big(\tau+(1-\tau)e^{-z/2}\big)$ and
\eqref{eq:illustration-winners} imply  $\pi_\C^\B\ge\xi(\lambda)(\pi_\C^\B+\pi_\C^\G)$.
Likewise, $H_\G\ge H_\B$ implies $\pi_\O^\G\ge(1-\beta)(\pi_\O^\B+\pi_\O^\G)$.
Since 
$\pi_\C^\B+\pi_\C^\G = 1-e^{-m_\C} \ge 1-e^{-x(\lambda)}$ and $\pi_\O^\B+\pi_\O^\G = 1-e^{-m_\O} \ge 1-e^{-y(\lambda)}$,
\begin{equation*}
\begin{cases} \pi_\C^\B \ge \xi(\lambda)(1-e^{-x(\lambda)}) \\ \pi_\O^\B \le \beta(1-e^{-m_O}) \le \beta \end{cases} \Rightarrow\ \ 
\pi_\C^\B-\pi_\O^\B \ge\xi(\lambda)(1-e^{-x(\lambda)})-\beta=:\Delta_\B(\lambda). 
\end{equation*}
Similarly, 
\begin{equation*}
\begin{cases} \pi_\C^\G \le 1-\xi(\lambda) \\ \pi_\O^\G \ge (1-\beta)(1-e^{-y(\lambda)}) \end{cases}\!\Rightarrow\ \ 
\pi_\O^\G-\pi_\C^\G \ge(1-\beta)(1-e^{-y(\lambda)})-1+\xi(\lambda)=:\Delta_\G(\lambda). \label{eq:illustration-gap-bounds2}
\end{equation*}
Since $x$, $y$, $\underline z$, and $\xi$ are nondecreasing in $\lambda$,
so are $\Delta_\B$ and $\Delta_\G$. Strict positivity at $\lambda=\lambda_*$
therefore suffices for all $\lambda\ge\lambda_*$.

\subsection{Methodology and Verification}\label{app:simulation:verif}

We had ChatGPT approximate the values in \cref{tab:illustration-results} according to \cref{app:simulation:bound}. We checked each such output value from ChatGPT using the verification algorithm given in pseudocode in \cref{alg:lambda_star}, implemented in about 250 lines of Python as a floating-point approximation using \texttt{scipy}'s \texttt{brentq} for finding fixed points in pseudocode step (4). The verification code was produced pseudocode step by pseudocode step using Google Gemini, and manually verified line by line for adherence to the pseudocode specification. We provide the verification code in full at
\ifanon
    \url{https://anonymous.4open.science/r/security-by-obscurity-gt-3D54}.
\else
    \url{https://codeberg.org/nooske/security-by-obscurity-gt}.
\fi

\begin{algorithm}
\caption{Verification Algorithm for $\lambda^*$}
\label{alg:lambda_star}
\begin{itemize}
\item[(1)] $P_*(m)$ and $P^*(m)$ are calculated as below (per \cref{fn:poisson-def}).
\[
P_*(m) = e^{-m}\sum_{n=0}^{\infty}
\frac{m^n}{n!}\frac{1}{1+\gamma n} \qquad\mbox{and}\qquad
P^*(m) = e^{-m}\sum_{n=0}^{\infty}
\frac{m^n}{n!}\frac{\gamma}{\gamma+n}\ . \]

\item[(2)] Define $Q_\tau(d)$ as \eqref{eq:illustration-response}.

\item[(3)] Define $R_\sys(v|\lambda)$ as \eqref{eq:illustration-load-maps} for $v\in(0,\lambda)$ given $\lambda>0$.

\item[(4)] Define $x(\lambda)$, $y(\lambda)$ as the $v$ that solves $R_\C(v;\lambda)=v$, $R_\O(v;\lambda)=v$ respectively.

\item[(5)] Define $\underline z:=k/P^*(x(\lambda))$ and $\xi(\lambda)$ as \eqref{eq:illustration-xi}.

\item[(6)] Define $\Delta_\B(\lambda):=\xi(\lambda)(1-e^{-x(\lambda)})-\beta$ and
		$\Delta_\G(\lambda):=(1-\beta)(1-e^{-y(\lambda)})-1+\xi(\lambda)$.
		
\item[(7)] Verify that $\Delta_\B(\lambda^*)>0$ and $\Delta_\G(\lambda^*)>0$ for the parameter values in \cref{tab:illustration-results}.

\end{itemize}
\end{algorithm}

\end{document}